\documentclass[aps,prl,reprint,superscriptaddress,longbibliography]{revtex4-2}

\usepackage{amsmath,amssymb,amsthm}
\usepackage{bm}
\usepackage[scr=boondox,  % heavily sloped
            cal=esstix]   % slightly sloped
           {mathalpha}
\usepackage{graphicx}
\usepackage{xcolor}
\usepackage[colorlinks=true, linkcolor=blue, citecolor=blue, urlcolor=blue]{hyperref}
\usepackage{ulem}

\newcommand{\ii}{\mathrm{i}}

\newcommand{\Tr}{\operatorname{Tr}}

\newtheorem{theorem}{Theorem}
\newtheorem{corollary}[theorem]{Corollary}

\begin{document}

\title{Scalable simulation of non-Markovian quantum transport by stochastic-phase bath reduction}

\author{Zhen Huang}
\thanks{Contact author: zhuang@flatironinstitute.org}
\affiliation{Department of Mathematics, University of California, Berkeley, CA 94720, USA}
\affiliation{Center for Computational Quantum Physics, Flatiron Institute, New York, NY 10010, USA}
\affiliation{Center for Computational Mathematics, Flatiron Institute, New York, NY 10010, USA}
\author{Lin Lin}
\thanks{Contact author: lin@caltech.edu}
\affiliation{Department of Mathematics, University of California, Berkeley, CA 94720, USA}
\affiliation{Applied Mathematics and Computational Research Division, Lawrence Berkeley National Laboratory, Berkeley, CA 94720, USA}
\affiliation{Department of Computing and Mathematical Sciences, California Institute of Technology, Pasadena, CA 91125, USA}
\author{Pinchen Xie}
\thanks{Contact author: pinchenxie@lbl.gov}
\affiliation{Applied Mathematics and Computational Research Division, Lawrence Berkeley National Laboratory, Berkeley, CA 94720, USA}

\date{\today}

\begin{abstract}
Accurate simulations of non-Markovian quantum transport remain limited to small systems, and a key bottleneck is that each system site requires an independent set of quantum bath orbitals. We introduce the Stochastic Phase Algorithm (SPA), which replaces the independent local Gaussian baths in single-quasiparticle Holstein models by $R$ shared quantum baths with stochastic site-dependent phases. We prove that phase averaging reproduces the target two-point bath-correlation matrix for every $R$, while increasing $R$ systematically suppresses errors from higher-order cross-site correlations. We obtain a fixed-time $O(R^{-1})$ trace-norm error bound with a prefactor independent of system size and connectivity. This spatial compression applies directly to unitary baths and can be combined with coupled-Lindblad spectral compression [Phys. Rev. Lett. 136, 090403 (2026)], which represents each shared environment by a very small number of coupled, damped quantum modes fitted to its thermal bath correlation over the simulation window. Numerically, SPA converges against direct unitary benchmarks on a $3\times3$ lattice, enables full-state-vector quantum-bath dynamics on a $100\times100$ lattice, and, with $R=1$, closely reproduces near-exact benchmarks for exciton populations in bacteriochlorophyll aggregates and carrier mobility in rubrene.
\end{abstract}

\maketitle

The interplay between electronic and nuclear degrees of freedom governs the transport of quasiparticles in molecular systems, including polarons in organic semiconductors \cite{Coropceanu2007ChemRev,Troisi2011ChemSocRev} and excitons in photosynthetic complexes \cite{Ishizaki2009PNAS,Scholes2011NatChem}. Nuclear quantum effects (NQE) and electron-vibronic coupling (EVC) render such transport intrinsically an open quantum dynamical problem, motivating computational methods that retain environmental memory beyond Markovian master equations for the reduced electronic system.

Compared with generic open quantum dynamics, quantum transport in extended molecular systems is especially challenging for two reasons. First, each electronically active site can be coupled to its own structured vibronic environment. The resulting Hilbert-space dimension grows exponentially with system size~\footnote{Here, the Hilbert-space dimension refers to the direct-product representation. Tensor-network methods may compress this space when the relevant state or influence functional has a favorable low-rank or entanglement structure.}. Second, EVC is often comparable to coherent tunneling, placing transport in the intermediate-coupling regime where neither weak- nor strong-coupling theories are generally controlled~\cite{Coropceanu2007ChemRev, Troisi2011ChemSocRev,Fratini2016, oberhofer2017charge,shuai2020applying}.

Extended systems with intermediate EVC thus provide a stringent testbed for non-perturbative open-system methods~\footnote{We refer readers to Ref.~\cite{lacroix2026tensor} for a comprehensive review of non-perturbative methods.}. The cost of systematically improvable methods, including influence-functional~\cite{Makri1995QUAPI,Strathearn2018TEMPO,ye2021constructing}, hierarchy~\cite{Tanimura2020HEOM}, and effective-environment approaches~\cite{Prior2010TEDOPA,Tamascelli2018Pseudomodes, DAMPF2019}, grows with environmental memory, bath complexity, and system-bath entanglement~\cite{lacroix2026tensor}. Bath compression~\cite{thoenniss2025efficient,huang2026coupled}, tensor-network compression~\cite{Prior2010TEDOPA,DAMPF2019,ye2021constructing}, and adaptive localization~\cite{varvelo2021formally,varvelo2023formally} can extend accessible system sizes, yet controlled simulations of long-time, delocalized transport on large multidimensional networks remain challenging. Mixed quantum-classical methods achieve greater scalability by representing the environments semiclassically~\cite{ehrenfest1927bemerkung,Tully1990FSSH,wang2011mixed,jiang2016nuclear, Xie2026DIQCD}. However, they do not explicitly retain system-bath entanglement, generally lack systematic convergence to exact quantum dynamics, and may compromise detailed balance~\cite{Parandekar2006DetailedBalance,gu2019can,Amati2023DetailedBalance}.

In this Letter, we introduce the Stochastic Phase Algorithm (SPA)~\footnote{ We use this term to distinguish it from the well-known random-phase approximation (RPA) in perturbation theory.}, a systematically improvable method that reduces the spatial multiplicity of vibronic environments while retaining quantum system-bath dynamics. We consider single-quasiparticle transport on an $N_{\text{s}}$-site tight-binding network of arbitrary connectivity, where each site is coupled to an identical but independent Gaussian bosonic bath [Fig.~\ref{fig:overview}]. This homogeneous environment defines a Holstein-type model~\cite{Holstein1959PartI,Holstein1959PartII}, widely used for organic materials composed of equivalent or similar molecular units.

Unlike existing tensor-network approaches~\cite{Prior2010TEDOPA, DAMPF2019} that retain site-resolved environments, SPA replaces the $N_{\text{s}}$ local baths with a small number ($R$) of shared quantum baths dressed by static, stochastic site-dependent phases. Each realization retains the local bath spectrum and supports system-bath entanglement, while phase averaging preserves complete positivity. We prove that phase averaging reproduces the full matrix of two-point bath correlation functions (BCFs) for every $R$, while errors in higher-order cross-site correlations are suppressed as $1/R$.
We derive a fixed-time $O(R^{-1})$ trace-norm error bound for the phase-averaged reduced density operator, with a prefactor independent of $N_{\text{s}}$ and network connectivity. Thus, for a given time and accuracy, enlarging the network does not require more bath copies.

SPA can also incorporate pseudomode methods for spectral bath compression~\cite{park2024quasi,huang2026provably,thoenniss2025efficient}.
Here, we use the recently developed coupled-Lindblad compression~\cite{huang2026coupled}, which represents each shared bath by a few damped quantum modes while retaining BCFs. SPA reduces the number of spatial bath copies, while coupled-Lindblad compression reduces the modes within each copy.
Their combination connects first-principles or experimentally validated bath models to mesoscale transport phenomena in organic electronics and photosynthetic assemblies on scales previously accessible mainly with semiclassical methods~\cite{giannini2022exciton} or quantum-bath methods whose favorable scaling relies on quasiparticle localization, and whose mesoscale demonstrations have been limited to overdamped vibronic environments~\cite{varvelo2021formally,varvelo2023formally}.

\begin{figure}[t]
    \centering
    \includegraphics[width=0.85\linewidth]{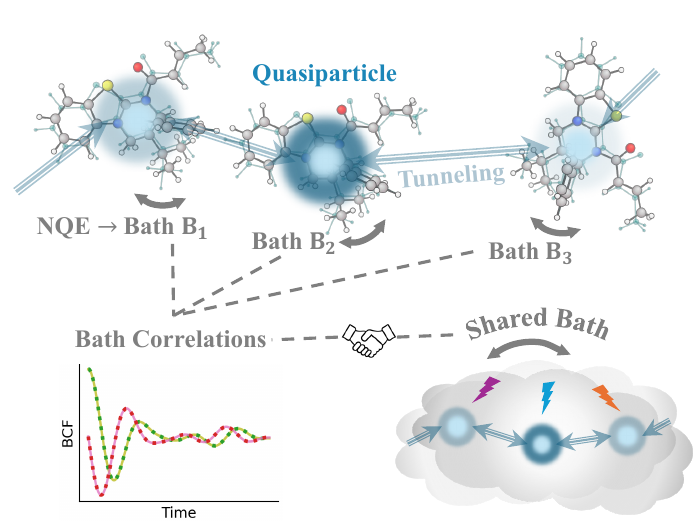}
    \caption{Schematics of quantum transport in organic materials and the SPA approach.}
    \label{fig:overview}
\end{figure}

\vspace{1em}
\textit{Model.--}
We focus on quantum transport described by an $N_{\text{s}}$-site Holstein Hamiltonian, $\hat H=\hat H_{\text{s}}+\hat H_{\text{b}}+\hat H_{\text{sb}}$. 
The tight-binding Hamiltonian is 
$\hat H_{\text{s}}=\sum_{i=1}^{N_{\text{s}}}U_i |i\rangle\langle i|+\sum_{\substack{i,j=1\\i\neq j}}^{N_{\text{s}}}V_{ij}|i\rangle\langle j|$,
where $U_i$ are site energies and $V_{ji}=V_{ij}^*$ are hopping amplitudes on a network of arbitrary connectivity.
The vibronic environment of each site consists of $n$  bosonic modes, $\hat H_{\text{b}}=\sum_{i=1}^{N_{\text{s}}}\sum_{\alpha=1}^{n}\omega_{\alpha}\hat b_{i\alpha}^\dagger \hat b_{i\alpha}$, where $\hat b_{i\alpha}^\dagger$ creates a vibronic excitation in mode $\alpha$ attached to site $i$. The system-bath coupling is local and diagonal,
\begin{small}
\begin{equation}\label{eq:ham-eb}
    \hat H_{\text{sb}}=\sum_{i=1}^{N_{\text{s}}}  |i\rangle\langle i| \otimes \sum_{\alpha=1}^{n}g_{\alpha}\omega_{\alpha}(\hat b_{i\alpha}^\dagger+\hat b_{i\alpha}),
\end{equation}
\end{small}with dimensionless coupling $g_{\alpha}$.
The local reorganization energy $\lambda=\sum_{\alpha=1}^{n}g_\alpha^2\omega_\alpha$ measures the energetic stabilization due to EVC. 
The intermediate-coupling regime occurs when $\lambda$ is comparable to the dominant hopping amplitude. In the Supplemental Material (SM)~\cite{SupplementalMaterial}, we list typical materials that fall into this regime~\cite{Coropceanu2007ChemRev, Yamamoto2010DNTT, Pandey2023BTBT, Gajdos2013PCBM, Ishizaki2009PNAS, kundu2020real, Zhu2025LHCII, Blau2018PC645, Raszewski2008PSII, Kreisbeck2016PSII}.

We set $\hbar=1$ and $\beta=(k_{\mathrm B}T)^{-1}$.
We assume the initial state is a tensor product of a system state $\hat\rho_{\text{s}}(0)$ and a thermal environmental state $\hat\rho_{\text{b},\beta}\propto\exp(-\beta \hat H_{\text{b}})$. We denote the exact reduced system density operator by $\hat\rho_{\text{s}}(t)$. The environmental influence on $\hat\rho_{\text{s}}(t)$ is encoded in the BCF $C_{ij}(t)$ between sites $i$ and $j$, which can be obtained from first-principles calculations or from experimental vibronic spectral densities~\cite{Giustino2017,Valleau2012,Pachon2014}.
For independent Gaussian baths, $C(t)$ is diagonal:
$C_{ij}(t)=c(t)\delta_{ij}$, 
with
$c(t)=\sum_{\alpha=1}^{n} g_\alpha^2\omega_\alpha^2
[\coth(\beta\omega_\alpha/2)\cos(\omega_\alpha t)-\ii\sin(\omega_\alpha t)]$.
For a continuous bath with spectral density $J(\omega)$, this discrete sum is replaced by an integral over $\omega$. 

In this work, we focus on reduced system observables
$O(t) = \Tr[\hat\rho_{\text{s}}(t)\hat O]$.
Of particular interest is the quasiparticle population on site $i$,
$P_i(t) = \Tr[\hat\rho_{\text{s}}(t)\,|i\rangle\langle i|]$.

\begin{figure*}[bt]
    \centering
    \includegraphics[width=0.9\linewidth]{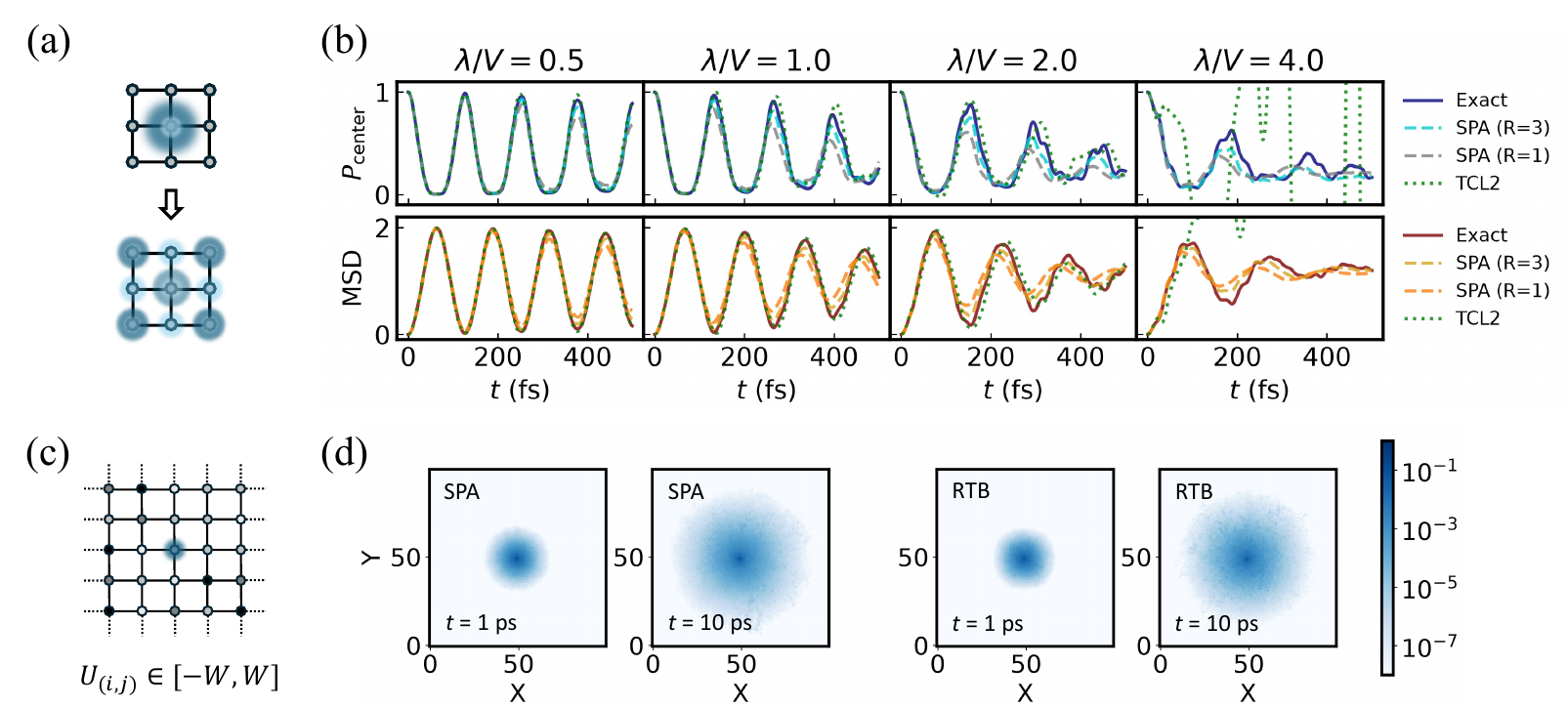}
    \caption{(a) Schematic of the population dynamics on the $3\times 3$ lattice. (b) $P_{\text{center}}$ and the MSD as functions of time, obtained from exact simulation, SPA, and TCL2. The statistical uncertainties in the SPA results are smaller than the plotted line width and are therefore not shown. (c) Schematic of the disordered site energies $U_{(X,Y)}$ and the localized initial state on the 2D lattice. $U_{(X,Y)}\in[-W,W]$ is uniformly sampled with $W=200~\mathrm{cm}^{-1}$ for each realization of random phases. (d) Snapshots of the site populations $P_{(\text{X,Y})}(t)$ on a logarithmic color scale obtained with $\lambda/V=8.0$.}
    \label{fig:method}
\end{figure*}

\vspace{1em}
 \textit{Method.--}
Instead of retaining $n$ bosonic modes per site, SPA creates $R$ copies of the local bath and couples them to all sites. The shared baths are given by
$\hat H_{\text{b}}^{(R)}=\sum_{a=1}^{R}\sum_{\alpha=1}^{n}
\omega_\alpha \hat b_{a\alpha}^\dagger \hat b_{a\alpha}$, 
with the phase-dressed coupling
\begin{small}
\begin{equation}\label{eq:ham-eb-eff}
\hat H_{\text{sb},\theta}^{(R)}
=
\sum_{i=1}^{N_{\text{s}}} |i\rangle\langle i|
\otimes
\sum_{a=1}^{R}\sum_{\alpha=1}^{n}
\frac{g_\alpha\omega_\alpha}{\sqrt R}
\left(
r_i^{(a)}\hat b_{a\alpha}^\dagger
+
r_i^{(a)*}\hat b_{a\alpha}
\right).
\end{equation}
\end{small}
Here $r_i^{(a)}=e^{\mathrm{i}\theta_i^{(a)}}$, with the phases $\theta_i^{(a)}$ sampled independently and uniformly from $[0,2\pi)$. The total Hamiltonian,
$\hat H_\theta^{(R)}=\hat H_{\text{s}}+\hat H_{\text{b}}^{(R)}+\hat H_{\text{sb},\theta}^{(R)}$, 
generates unitary dynamics, $\frac{\mathrm{d}}{\mathrm{d}t}\hat\rho_\theta^{(R)}(t)=-\mathrm{i} [\hat H_\theta^{(R)},\hat\rho_\theta^{(R)}(t)]$, with each shared bath initialized in the thermal state. Let $\hat\rho_{\text{s},\theta}^{(R)}(t)=\Tr_{\text{b}^{(R)}}\hat\rho_\theta^{(R)}(t)$ be the reduced state for a fixed phase, and $\hat\rho_{\text{s}}^{(R)}(t)=\mathbb E_\theta[\hat\rho_{\text{s},\theta}^{(R)}(t)]$ its phase average.
The phase average restores the two-point BCFs of the original independent baths. Specifically,
$\mathbb E_\theta\ [
R^{-1}\sum_{a=1}^{R}
r_i^{(a)}r_j^{(a)*}
]=\delta_{ij}$,  and hence
$\mathbb E_\theta\ [C_{ij,\theta}^{(R)}(t)] =c(t)\delta_{ij}$.

Thus, SPA always reproduces the two-point BCFs, but this alone does not guarantee exact dynamics because the phase-averaged Gaussian baths need not remain Gaussian. For example, we define $G_{ij}^{(R)}=R^{-1}\sum_{a=1}^R r_i^{(a)}r_j^{(a)*}$. For $i\neq j$, we have $\mathbb E_\theta[G_{ij}^{(R)}]=0$, but $\mathbb E_\theta[G_{ij}^{(R)}G_{ji}^{(R)}]=1/R$, which contributes to fourth-order cross-site errors. Our main theoretical result is that (see the End Matter) the error in SPA is systematically controlled as
\begin{equation}
\left\|\hat\rho_{\mathrm s}(t)-\hat\rho_{\mathrm s}^{(R)}(t)\right\|_{\mathrm{tr}} \leq A_t / R.
\end{equation}
Here $A_t$ depends on time and the BCFs, but not on $N_{\text{s}}$ or connectivity. The error in a bounded system observable $\hat O$ is at most $\|\hat O\|_\infty A_t/R$. This analysis takes all high-order errors into consideration. At fixed time and target accuracy, the number of shared baths required by SPA is independent of system size.
However, this size-independent bound generally does not extend to many-body settings where the relevant electronic sector grows with system size.

So far, SPA has reduced the spatial multiplicity of the environments but not the number of modes $n$ in a shared bath. When $n$ is large, or when the shared bath has a continuous spectral density, further dimension reduction~\cite{park2024quasi,huang2026coupled,huang2026provably,muller2026one} is needed. This separate bottleneck is addressed by coupled-Lindblad compression, which replaces an $n$-mode unitary bath by $n_{\mathrm b}$ damped quantum modes that exchange excitations~\cite{huang2026coupled}. For continuous baths at fixed BCF accuracy, the mode count in standard unitary discretizations grows linearly with the simulation duration, whereas coupled-Lindblad compression can achieve polylogarithmic growth. We fit a bath Hamiltonian matrix $K=K^\dagger$, a damping matrix $\Gamma=\Gamma^\dagger\succeq0$, and a coupling vector $\epsilon$ so that $\tilde c(t)=\epsilon^\dagger e^{(-\mathrm{i}K-\Gamma)t}\epsilon$ approximates $c(t)$ for $0\leq t\leq\tau$~\footnote{All coupled-Lindblad fits reported in this work were performed using \texttt{RealTimeBath} package~\cite{RealTimeBath}.}. The auxiliary modes are initialized in their vacuum, with the physical temperature encoded in $K, \Gamma$ and $\epsilon$. Although the enlarged system obeys a Lindblad equation, the retained modes carry memory and the reduced electronic dynamics remains non-Markovian.
SPA naturally adapts to this reduction, and the shared-bath Hamiltonian becomes $\hat H_{\text{b}}^{(R)}=\sum_{a=1}^{R}\sum_{k,l=1}^{n_{\text{b}}}
K_{kl}\hat b_{ak}^\dagger\hat b_{al}$.
The system-bath coupling retains the form of Eq.~\eqref{eq:ham-eb-eff} under $n\rightarrow n_{\text{b}}$ and $\alpha\rightarrow k$, with $\epsilon_k$ multiplying $\hat b_{ak}^\dagger$ and $\epsilon_k^*$ multiplying $\hat b_{ak}$. The phases $\theta_i^{(a)}$ are sampled in the same way. The resulting dynamics are governed by the Lindblad equation
$\frac{\mathrm{d}}{\mathrm{d}t}\hat\rho_{\theta}^{(R)}(t)=-\mathrm{i}[\hat H_{\theta}^{(R)},\hat\rho_{\theta}^{(R)}(t)]+\mathcal D(\hat\rho_{\theta}^{(R)}(t))$,
where~\footnote{In the coupled-Lindblad convention, $\Gamma$ differs by a factor of two from the standard Lindblad rate, so that $\tilde c(t)=\epsilon^\dagger e^{(-\mathrm{i}K-\Gamma)t}\epsilon$ contains no additional factor of two in the damping term.}
\begin{small}
\begin{equation}
\mathcal D(\hat\rho)
=
\sum_{a=1}^{R}\sum_{k,l=1}^{n_{\text{b}}}
\Gamma_{kl}
\left(
2\hat b_{al}\hat\rho\hat b_{ak}^\dagger
-
\{\hat b_{ak}^\dagger\hat b_{al},\hat\rho\}
\right).
\end{equation}
\end{small}
We can simulate the Lindblad dynamics using the Monte Carlo wavefunction propagator~\cite{PlenioKnight1998} and average jointly over Monte Carlo trajectories and static phases~\cite{SupplementalMaterial}.

Finally, as an example of SPA’s scaling advantage in practical simulations, suppose that each bosonic mode is truncated to $d$ levels. The original system–bath state dimension is $N_{\text{s}} d^{nN_{\text{s}}}$. SPA reduces it to $N_{\text{s}} d^{nR}$ for unitary shared baths, and to $N_{\text{s}} d^{n_{\mathrm b}R}$ after spectral compression. The numerical benchmarks below show that $R=1$ can already yield accurate transport observables.

\begin{figure*}[t]
    \centering
    \includegraphics[width=0.9\linewidth]{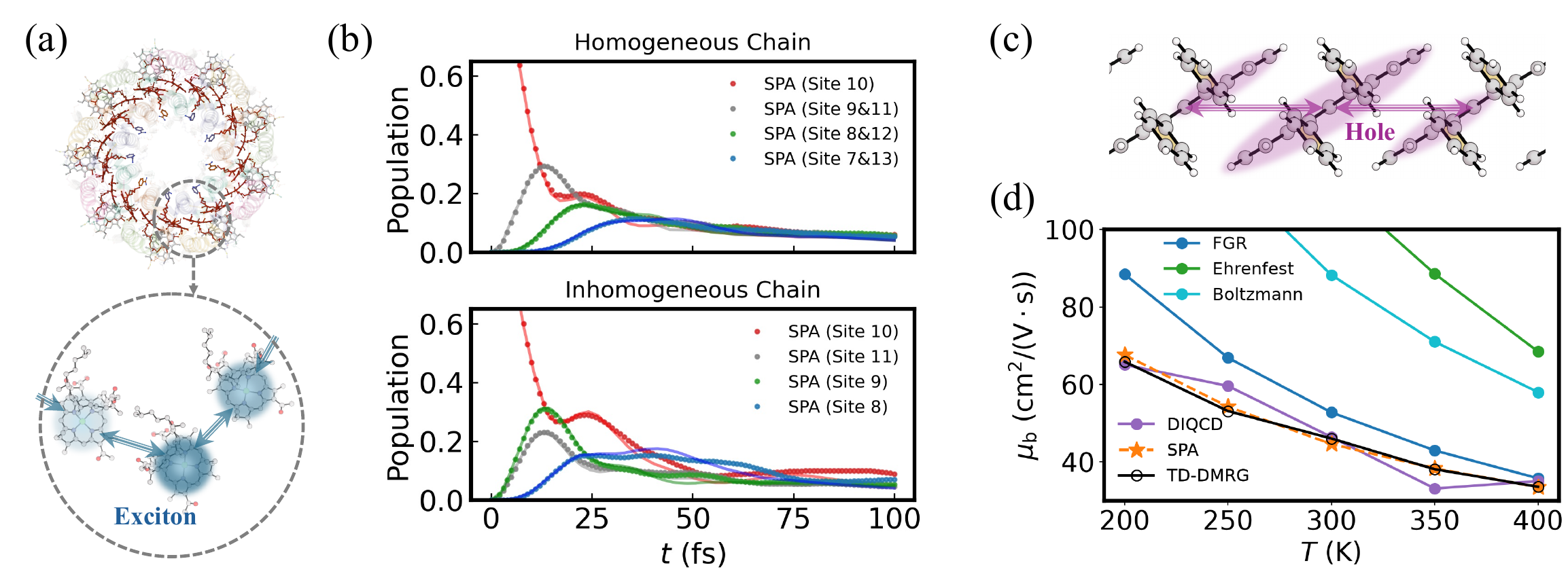}
    \caption{(a) The experimental crystal structure of LH2 B800-850~\cite{papiz2003structure,pdb1nkz} and a schematic of exciton migration within the B850 ring.  (b) Exciton population dynamics for homogeneous and inhomogeneous (staggered) 19-site BChl chains at $T=300$ K. SPA predictions are reported as dots, while MPI predictions are reported as lines.
    (c) The $bc$ plane of a rubrene crystal. Purple shading indicates a localized hole~\cite{ishii2017charge}. (d) $b$-axis carrier mobility as a function of temperature. In (b) and (d), statistical uncertainties of the SPA results are smaller than the marker size and are not shown.
    }
    \label{fig:app}
\end{figure*}

\textit{Systematic convergence and scalability in a 2D Holstein model.--}
We demonstrate SPA with a 2D Holstein model on an open $L\times L$ square lattice with nearest-neighbor hopping $V=100~\mathrm{cm}^{-1}$. Each site is coupled to a bosonic mode with $\omega_1=1000~\mathrm{cm}^{-1}$, and the reorganization energy $\lambda=g_1^2\omega_1$ is varied through $g_1$.

We demonstrate the accuracy of SPA on an $L\times L$ lattice with $L=3$ [Fig.~\ref{fig:method}(a)], which allows exact unitary benchmarking. We assume uniform site energy, $T=300$ K, and an initial state localized at the central site. For several values of $\lambda$, we compute the central-site population $P_{\mathrm{center}}$ and the mean-squared displacement,
$\mathrm{MSD}(t)=\Tr[\hat\rho_{\text{s}}(t) (\hat X^2+\hat Y^2)]-\bigl(\Tr[\hat\rho_{\text{s}}(t)\hat X]\bigr)^2-\bigl(\Tr[\hat\rho_{\text{s}}(t)\hat Y]\bigr)^2$, where $\hat X$ and $\hat Y$ are the site-coordinate operators. Fig.~\ref{fig:method}(b) compares the exact results with SPA and second-order time-convolutionless perturbation theory (TCL2)~\cite{Fetherolf2017TCL2}. The SPA error grows with $\lambda$ but decreases systematically with $R$. TCL2 is comparably accurate only at weaker coupling; for $\lambda/V>2$, it breaks down and produces nonphysical dynamics that violate positivity. In contrast, the completely positive SPA is robust throughout the intermediate-coupling regime.

For scalability, we simulate a particle initially localized at the center of an $L\times L$ lattice with $L=100$ and disordered site energies [Fig.~\ref{fig:method}(c)].
Because no systematically converged benchmark is feasible at this scale, we compare SPA ($R=1$) against a renormalized tight-binding (RTB) model that incorporates EVC through the Lang--Firsov hopping $V'=V\exp[-g_1^2\coth(\beta\omega_1/2)]$. We choose $\lambda/V=8$, for which the RTB model provides a qualitative strong-coupling reference. The population snapshots in Fig.~\ref{fig:method}(d) show qualitative agreement between SPA and RTB. Both models exhibit slow diffusion over $10~\mathrm{ps}$, resulting from strong-coupling suppression of transport and disorder frustration.
For comparison, state-of-the-art tensor-network simulations~\cite{DAMPF2019} that retain site-resolved environments have demonstrated long-time dynamics only for tens of sites.

These benchmarks suggest that SPA with $R=1$ is a practical choice for balancing accuracy and efficiency. It enables full-state-vector simulations of realistic mesoscale systems. Realistic vibronic environments, however, are multimodal and substantially more complex than the single-mode bath considered above. The ratio $\lambda/V$ alone is insufficient to predict the accuracy of SPA. We therefore assess SPA with $R=1$ for transport in representative organic materials.

\textit{Exciton migration in a photosynthetic complex.--}
Frenkel-exciton migration in photosynthetic antenna complexes is commonly described by Holstein models, making bacteriochlorophyll (BChl) aggregates a natural testbed for SPA. We consider the LH2 complex, which contains 9 BChl \textit{a} molecules in the B800 ring and 18 in the B850 ring [Fig.~\ref{fig:app}(a)]. The B850 ring is the primary channel for exciton transport.

For benchmarking, we use the homogeneous and staggered 19-site B850-chain models of Ref.~\cite{kundu2020real}, for which near-exact modular path-integral (MPI) results are available. Each site couples to 50 vibrational modes parameterized for \textit{Cereibacter sphaeroides}~\footnote{See Table II of Ref.~\cite{ratsep2011demonstration}; the parameters are reproduced in SM~\cite{SupplementalMaterial}}, with $\lambda=217.66~\mathrm{cm}^{-1}$; the Hamiltonian parameters are listed in the SM~\cite{SupplementalMaterial}. The exciton is initially localized at site 10.

For $T=300~\mathrm{K}$, we compress the 50-mode local bath into a shared, six-mode dissipative bath~\cite{SupplementalMaterial} with $\tau=100~\mathrm{fs}$, matching the timescale studied in Ref.~\cite{kundu2020real}.   Figure~\ref{fig:app}(b) compares the resulting $R=1$ SPA populations with the path integral approach.  SPA captures coherent spreading and vibrational damping in the homogeneous chain, as well as energy-biased redistribution in the staggered chain. Thus SPA retains both coherent and relaxational features of a structured, non-Markovian molecular bath. The bath-compression error is negligible over this time window, and the small deviations beyond $t\approx30~\mathrm{fs}$ are consistent with the finite-$R$ error.

\textit{Charge transport in an organic semiconductor.--}
Charge transport in high-mobility organic semiconductors often lies in the intermediate-coupling regime. The time-dependent density matrix renormalization group (TD-DMRG) method provides near-exact benchmarks for directional carrier mobilities in Holstein-type models~\cite{li2020finite,li2021general}, but its computational cost restricts simulations to small 1D chains~\cite{li2020finite} or quasi-2D ladders~\cite{li2021general}.
 
We benchmark SPA against TD-DMRG for hole transport along the high-mobility $b$ axis of rubrene [Fig.~\ref{fig:app}(c)]. We use the homogeneous 1D Holstein model of Refs.~\cite{jiang2016nuclear,li2020finite,Xie2026DIQCD}, parameterized using hybrid density-functional theory~\cite{jiang2016nuclear}. The nearest-neighbor hopping is $V=83~\mathrm{meV}$, and each site couples to nine vibrational modes [see SM~\cite{SupplementalMaterial}], giving $\lambda=73~\mathrm{meV}$. Bath compression is not necessary.

While TD-DMRG is restricted to equilibrium linear-response calculations for a few dozen sites, SPA can directly simulate unitary transport on a chain with $L=200$ sites. The hole is initially localized at the central site. The $b$-axis mobility is obtained from
$\mu_{\text{b}}(T)=\frac{eD^2}{2k_{\mathrm B}T}\lim\limits_{t\to\infty}
\frac{\mathrm{d}\,\mathrm{MSD}(t)}{\mathrm{d}t}$~\footnote{The asymptotic slope is extracted from the diffusive regime before finite-size effects.},
with intermolecular spacing $D=7.19~\text{\AA}$.  
In Fig.~\ref{fig:app}(d), we compare the $R=1$ SPA prediction for $\mu_{\text{b}}(T)$ against predictions from TD-DMRG~\cite{li2020finite}, Fermi's golden rule (FGR)~\cite{li2020finite}, Boltzmann theory~\cite{li2020finite}, Ehrenfest dynamics, and data-informed quantum-classical dynamics (DIQCD)~\cite{Xie2026DIQCD}~\footnote{The setup of Ehrenfest and DIQCD simulations in this work differs from that in Ref.~\cite{Xie2026DIQCD}, where shorter propagation times, smaller bosonic cutoffs, and $D=7~\text{\AA}$ led to overestimated mobilities. With these parameters corrected, DIQCD agrees substantially better with TD-DMRG than previously reported.}.
Because $\lambda/V\simeq0.88$ places rubrene in the intermediate-coupling regime, neither weak-coupling Boltzmann theory nor strong-coupling FGR is reliable. In contrast, SPA closely reproduces the TD-DMRG mobility from $200$ to $400~\mathrm{K}$, whereas the Ehrenfest results exhibit systematic deviations. DIQCD is a semiclassical method that uses machine learning to reproduce selected quantum-bath statistics. Its prediction is close to the TD-DMRG result but exhibits nonmonotonic residual errors, consistent with the uncontrolled nature of the semiclassical approximation. This comparison therefore validates SPA for modeling long-time transport in the regime where coherent tunneling and strong NQE compete.

Beyond this case study, organic systems can exhibit strongly anisotropic tunneling that requires 2D or 3D transport models. Future work will explore such multidimensional transport.

\textit{Discussion.--}
SPA achieves systematically improvable accuracy by reducing the spatial multiplicity of quantum baths. Phase averaging restores the target BCFs, and the fixed-time finite-$R$ error bound is independent of system size. Thus, SPA avoids the exponential growth from replicated local bath Hilbert spaces. SPA preserves complete positivity, and the unitary thermal-bath construction recovers detailed balance in the Davies limit~\footnote{This follows by the standard weak-coupling argument, in which the thermal KMS relation imposes detailed balance on the Davies rates~\cite{Davies1974Markovian,Ding2025Efficient}.}.

Even with $R=1$, SPA accurately captures structured non-Markovian effects in BChl aggregates and reproduces the long-time carrier mobility of rubrene in the intermediate-coupling regime. Coupled-Lindblad compression reduces the number of modes per shared bath, providing a route to mesoscale transport simulations with structured quantum environments.

These bath reductions work with various representations of the system–bath state. We provide a Python package, \texttt{scalabath}~\cite{scalabath}, that implements GPU-parallel SPA simulations with a full state-vector representation and supports systems with thousands of molecular sites. Adopting a tensor-network representation~\cite{Vidal2004TEBD,WhiteFeiguin2004,Daley2004,Haegeman2011TDVP} of the system-bath state may provide further scalability. 

Beyond organic materials, SPA may model molecular-qubit arrays with spin-phonon coupling~\cite{lohaus2026experimental}. The framework may be extended to fermionic environments or to nonlocal EVC, including Holstein--Peierls models, through appropriately correlated stochastic phases and auxiliary-bath couplings.

\textit{Data and Code Availability.---}
All SPA simulations were conducted using the publicly available \texttt{scalabath} package hosted on GitHub~\cite{scalabath}. Coupled-Lindblad bath compression was performed with the \texttt{RealTimeBath} package, which is also publicly available on GitHub~\cite{RealTimeBath}. All results reported in this paper are reproducible using scripts that will be made publicly available via Zenodo.

\textit{Acknowledgments.---}
 We thank Mingyu Kang and Haoen Li for fruitful discussions. This work was supported in part by the Simons Targeted Grants in Mathematics and Physical Sciences on Moir\'e Materials Magic, and by the Applied Mathematics Program of the US Department of Energy (DOE) Office of Advanced Scientific Computing Research under contract number DE-AC02-05CH1123 (Z.H., L.L.). P. X. was supported by the Alvarez Fellowship of Lawrence Berkeley National Lab and DOE Advanced Scientific Computing Research (ASCR) Applied Mathematics program under Contract No. DE-AC02-05CH11231. This research used resources of the National Energy Research Scientific Computing Center (NERSC), a Department of Energy User Facility (project m5218-2026).

\nocite{Saad1992,Trotter1959,Suzuki1991,PlenioKnight1998}
\bibliography{references}
\clearpage  

\section*{End Matter}
\phantomsection
\label{sec:endmatter}
In this End Matter, we prove the system-size-independent $O(R^{-1})$ convergence of the phase-averaged SPA dynamics, first for unitary auxiliary baths and then for the coupled-Lindblad realization used in our simulations. Although the bosonic Hilbert spaces are infinite-dimensional, the Gaussian initial state guarantees absolute convergence of the series expansion used below~\cite{huang2024unified}, so the error bound is independent of the finite bosonic cutoff.

\begin{theorem}
For the exact Holstein model and the unitary $R$-bath SPA model defined in the main text, let $\hat\rho_{\text{s}}(t)$ and $\hat\rho_{\text{s}}^{(R)}(t)$ denote their respective reduced system states, with the latter averaged over the random phases. Then
\begin{equation}
\|\hat\rho_{\text{s}}(t)-\hat\rho_{\text{s}}^{(R)}(t)\|_{\rm tr} \leq  \frac{A_t}{R}
\end{equation}
where $A_t$ depends on $t$ and the BCF $c$ but is independent of $N_{\text{s}}$. For every bounded electronic observable \(O\), duality gives \(\left|\operatorname{Tr}[O(\rho_{\mathrm s}-\rho_{\mathrm s}^{(R)})]\right| \leq \|O\|_\infty A_t/R\).
\label{thm:main}
\end{theorem}

\begin{proof}
We first prove the zero-temperature case. At finite temperature, write $c^{\sigma\tau}=c_{+}^{\sigma\tau}+c_{-}^{\sigma\tau}$ and define $F_{t,\beta}:=\sum_{\sigma,\tau}\int_0^t\int_0^t (\lvert c_{+}^{\sigma\tau}(s,s')\rvert+\lvert c_{-}^{\sigma\tau}(s,s')\rvert)\,\mathrm{d}s'\,\mathrm{d}s$ and $A_{t,\beta}:=(F_{t,\beta}^2/4)\mathrm e^{F_{t,\beta}/2}$. The argument below then applies with $F_t$ and $A_t$ replaced by $F_{t,\beta}$ and $A_{t,\beta}$, respectively.
Let us define \(\hat P_i=|i\rangle\langle i|\), \(\hat U(t)=\mathrm e^{-\mathrm{i}\hat H_{\text{s}}t}\), and \(\hat S_i(t)=\hat U^\dagger(t)\hat P_i\hat U(t)\). Note that \(\sum_i\hat P_i=\hat I\). We use \(\mathcal S_i^{\mathrm l}(t)\hat X=\hat S_i(t)\hat X\) and \(\mathcal S_i^{\mathrm r}(t)\hat X=\hat X\hat S_i(t)\). We will use the notation \(\sigma_m\in\{\mathrm l,\mathrm r\}\) to record whether the \(m\)-th interaction acts by left or right multiplication on the reduced density operator; we refer to this as the branch label and set \(\eta_{\mathrm l}=1\), \(\eta_{\mathrm r}=-1\). Let \(\Delta_n(t)=\{0\leq t_n\leq\cdots\leq t_1\leq t\}\). For multi-indices \(\bm i\) and \(\bm\sigma\), write \(\mathcal S_{\bm i}^{\bm\sigma}(\bm t)=\mathcal S_{i_1}^{\sigma_1}(t_1)\cdots\mathcal S_{i_n}^{\sigma_n}(t_n)\), and let \(\sum_{\bm i,\bm\sigma}\) denote the sum over all site and branch labels.

On the bath side, we set up the notation for the unitary SPA realization. The coupled-Lindblad realization is addressed in Corollary~\ref{cor:lindblad}. Let \(\hat\rho_{\text{b}}\) and \(\hat\rho_{\text{b}}^{(R)}\) denote the corresponding initial Gaussian bath states, and define \(\hat B_i(t)=\sum_\alpha g_\alpha\omega_\alpha(\hat b_{i\alpha}^\dagger \mathrm e^{\mathrm{i}\omega_\alpha t}+\hat b_{i\alpha}\mathrm e^{-\mathrm{i}\omega_\alpha t})\). For \(R\) auxiliary bath channels, let \(\theta=\{\theta_i^{(a)}\}\), \(r_i^{(a)}=\mathrm e^{\mathrm{i}\theta_i^{(a)}}\), and
\(\hat B_i^{(R)}(t;\theta)=R^{-1/2}\sum_{a,\alpha}g_\alpha\omega_\alpha(r_i^{(a)}\hat b_{a\alpha}^\dagger\mathrm e^{\mathrm{i}\omega_\alpha t}+r_i^{(a)*}\hat b_{a\alpha}\mathrm e^{-\mathrm{i}\omega_\alpha t})\). In the interaction picture, the exact coupling has the form \(\sum_i\hat S_i(t)\otimes\hat B_i(t)\). Let us also define superoperators \(\mathcal B_i^{\mathrm l}(t)\hat X=\hat B_i(t)\hat X\), \(\mathcal B_i^{\mathrm r}(t)\hat X=\hat X\hat B_i(t)\), \(\mathcal B_{i,R}^{\mathrm l}(t;\theta)\hat X=\hat B_i^{(R)}(t;\theta)\hat X\), and \(\mathcal B_{i,R}^{\mathrm r}(t;\theta)\hat X=\hat X\hat B_i^{(R)}(t;\theta)\). The branch-resolved scalar bath correlations are defined by \(\Tr_{\text{b}}[\mathcal B_i^\sigma(t)\mathcal B_j^\tau(s)\hat\rho_{\text{b}}]=\delta_{ij}c^{\sigma\tau}(t,s)\), for \(\sigma,\tau\in\{\mathrm l,\mathrm r\}\). Here products of superoperators are composed from right to left. Note that \(c^{\mathrm l\mathrm l}(t,s)=c^{\mathrm r\mathrm l}(t,s)=c(t-s)\), while \(c^{\mathrm r\mathrm r}(t,s)=c^{\mathrm l\mathrm r}(t,s)=(c(t-s))^*\).

We write the ordered products as \(\mathcal B_{\bm i}^{\bm\sigma}(\bm t)=\prod_{m=1}^n\mathcal B_{i_m}^{\sigma_m}(t_m)\) and \(\mathcal B_{\bm i,R}^{\bm\sigma}(\bm t;\theta)=\prod_{m=1}^n\mathcal B_{i_m,R}^{\sigma_m}(t_m;\theta)\). Define the exact bath moments \(\Phi_{\bm i}^{\bm\sigma}(\bm t)=\Tr_{\text{b}}[\mathcal B_{\bm i}^{\bm\sigma}(\bm t)\hat\rho_{\text{b}}]\) and the phase-averaged random-bath moments \(\Phi_{\bm i,R}^{\bm\sigma}(\bm t)=\mathbb E_\theta\Tr_{\text{b}^{(R)}}[\mathcal B_{\bm i,R}^{\bm\sigma}(\bm t;\theta)\hat\rho_{\text{b}}^{(R)}]\).

We work with the interaction-picture reduced states
\(\hat\rho_{\text{s},I}(t)=\hat U^\dagger(t)\hat\rho_{\text{s}}(t)\hat U(t)\) and
\(\hat\rho_{\text{s},I}^{(R)}(t)=\hat U^\dagger(t)\hat\rho_{\text{s}}^{(R)}(t)\hat U(t)\).
The exact interaction-picture reduced state has the Dyson expansion
\begin{small}
\begin{align*}
\hat\rho_{\text{s},I}(t)
&=
\sum_{n=0}^{\infty}(-\mathrm{i})^n
\int_{\Delta_n(t)} \mathrm{d}\bm t
\sum_{\bm i,\bm\sigma}
\left(\prod_{m=1}^{n}\eta_{\sigma_m}\right)
\Phi_{\bm i}^{\bm\sigma}(\bm t)
\nonumber\\
&\quad\times
\mathcal S_{\bm i}^{\bm\sigma}(\bm t)\hat\rho_{\text{s}}(0).
\end{align*}
\end{small}
The expansion for \(\hat\rho_{\text{s},I}^{(R)}(t)\) is obtained by replacing \(\Phi_{\bm i}^{\bm\sigma}(\bm t)\) with \(\Phi_{\bm i,R}^{\bm\sigma}(\bm t)\); the system factor \(\mathcal S_{\bm i}^{\bm\sigma}(\bm t)\) and branch signs are unchanged. The \(n=0\) term is \(\hat\rho_{\text{s}}(0)\). Throughout this proof, \(\|\cdot\|_{\rm tr}\) denotes the trace norm.

Since the baths are Gaussian, odd orders vanish. At order \(2m\), let
\(\pi=\{\{p_\ell,q_\ell\}\}_{\ell=1}^m\)
range over all pairings, with \(p_\ell<q_\ell\). Set
\(c_\ell^{\bm\sigma}(\bm t)
=c^{\sigma_{p_\ell}\sigma_{q_\ell}}(t_{p_\ell},t_{q_\ell})\).
At zero temperature, orient each pair as \((u_\ell,v_\ell)\) so that
its phase factor is \(r_{i_{u_\ell}}r_{i_{v_\ell}}^*\), and define
\(G_{ij}^{(R)}=R^{-1}\sum_{a=1}^Rr_i^{(a)}r_j^{(a)*}\).
Wick's theorem gives
\begin{small}
\begin{equation}\label{eq:wick-expansion}
\begin{aligned}
\Phi_{\bm i}^{\bm\sigma}(\bm t)
&=\sum_\pi\prod_{\ell=1}^m
c_\ell^{\bm\sigma}(\bm t)\,
\delta_{i_{u_\ell}i_{v_\ell}},\\
\Phi_{\bm i,R}^{\bm\sigma}(\bm t)
&=\sum_\pi\mathbb E_\theta
\prod_{\ell=1}^m
c_\ell^{\bm\sigma}(\bm t)
G_{i_{u_\ell}i_{v_\ell}}^{(R)}.
\end{aligned}
\end{equation}
\end{small}
Thus the only replacement is
\(\delta_{ij}\mapsto G_{ij}^{(R)}\).

Let \(\mathcal A_R=\{1,\ldots,R\}^m\) denote the set of channel assignments \(\bm a=(a_1,\ldots,a_m)\). Expanding the product of Gram factors gives \(\prod_{\ell=1}^{m}G_{i_{u_\ell}i_{v_\ell}}^{(R)}=R^{-m}\sum_{\bm a\in\mathcal A_R}\prod_{\ell=1}^{m}r_{i_{u_\ell}}^{(a_\ell)}r_{i_{v_\ell}}^{(a_\ell)*}\). For the site part of a fixed pairing and branch pattern, define
\begin{small}
\begin{align*}
\mathcal T_{\pi}^{\bm\sigma}(\bm t)\hat X
&=
\sum_{\bm i}
\prod_{\ell=1}^{m}
\delta_{i_{u_\ell}i_{v_\ell}}\,
\mathcal S_{\bm i}^{\bm\sigma}(\bm t)\hat X,
\nonumber\\
\mathcal M_{\pi,\bm a,\theta}^{\bm\sigma}(\bm t)\hat X
&=
\sum_{\bm i}
\prod_{\ell=1}^{m}
r_{i_{u_\ell}}^{(a_\ell)}
r_{i_{v_\ell}}^{(a_\ell)*}\,
\mathcal S_{\bm i}^{\bm\sigma}(\bm t)\hat X.
\end{align*}
\end{small}
Define the phase-averaged random site map by \(\mathcal T_{\pi,R}^{\bm\sigma}(\bm t)=R^{-m}\sum_{\bm a\in\mathcal A_R}\mathbb E_\theta\mathcal M_{\pi,\bm a,\theta}^{\bm\sigma}(\bm t)\). Writing \(\delta\hat\rho_{\text{s},I}^{(2m)}(t)\) for the exact minus random \(2m\)-th order terms, the preceding definitions give
\begin{small}
\begin{align*}
\delta\hat\rho_{\text{s},I}^{(2m)}(t)
&=
(-\mathrm{i})^{2m}
\int_{\Delta_{2m}(t)}\mathrm{d}\bm t
\sum_{\bm\sigma}
\left(\prod_{j=1}^{2m}\eta_{\sigma_j}\right)
\nonumber\\
&\quad\times
\sum_{\pi}
\left(\prod_{\ell=1}^{m}c_\ell^{\bm\sigma}(\bm t)\right)
\left(
\mathcal T_{\pi}^{\bm\sigma}(\bm t)
-
\mathcal T_{\pi,R}^{\bm\sigma}(\bm t)
\right)\hat\rho_{\text{s}}(0),
\end{align*}
\end{small}
If \(a_1,\ldots,a_m\) are distinct, phase independence gives
\begin{equation*}
\mathbb E_\theta
\prod_{\ell=1}^m
r_{i_{u_\ell}}^{(a_\ell)}
r_{i_{v_\ell}}^{(a_\ell)*}
=
\prod_{\ell=1}^m
\delta_{i_{u_\ell}i_{v_\ell}},
\end{equation*}
and hence
\(\mathbb E_\theta
\mathcal M_{\pi,\bm a,\theta}^{\bm\sigma}
=\mathcal T_\pi^{\bm\sigma}\).
Thus only assignments containing a repeated channel contribute to
\(\mathcal T_\pi^{\bm\sigma}-\mathcal T_{\pi,R}^{\bm\sigma}\).
Their fraction satisfies
\begin{equation}\label{eq:collision-fraction}
p_{\mathrm{coll}}(m,R)
=1-\frac{(R)_m}{R^m}
\leq\frac{m(m-1)}{2R},
\end{equation}
where $(R)_m=R(R-1)\cdots(R-m+1)$, with $(R)_m=0$ for $m>R$.
It remains only to control the site sums in the maps appearing above. For each auxiliary channel \(a\) and \(\varepsilon=\pm1\), define the phase-dressed site sum \(\hat D_\varepsilon^{(a)}(t)=\sum_i(r_i^{(a)})^\varepsilon\hat S_i(t)\). This operator is unitary because it equals \(\hat U^\dagger(t)[\sum_i(r_i^{(a)})^\varepsilon\hat P_i]\hat U(t)\). Thus, for every \(\bm a\) and \(\theta\), the fully summed map \(\mathcal M_{\pi,\bm a,\theta}^{\bm\sigma}(\bm t)\) is a product of left and right multiplications by unitaries, and so is a trace-norm contraction. Phase averages preserve this contraction bound, and \(\mathcal T_{\pi,R}^{\bm\sigma}(\bm t)\) is an average of such phase-averaged maps over \(\bm a\). To bound the exact map, introduce $m$ independent auxiliary phase vectors $\bm\vartheta=(\vartheta^{(1)},\ldots,\vartheta^{(m)})$ and assign a distinct vector to each Wick block. With $\bm a_\star=(1,\ldots,m)$, phase averaging gives
\begin{equation*}
\mathcal T_{\pi}^{\bm\sigma}(\bm t)\hat X
=
\mathbb E_{\bm\vartheta}
\mathcal M_{\pi,\bm a_\star,\bm\vartheta}^{\bm\sigma}(\bm t)\hat X.
\end{equation*}
Because every map inside this average is a trace-norm contraction, convexity gives $\|\mathcal T_{\pi}^{\bm\sigma}(\bm t)\hat X\|_{\rm tr}\leq\|\hat X\|_{\rm tr}$. Averaging the difference of these contractions over the repeated-channel assignments gives
\begin{small}
\begin{align}
\left\|
\left(
\mathcal T_{\pi}^{\bm\sigma}(\bm t)
-\mathcal T_{\pi,R}^{\bm\sigma}(\bm t)
\right)\hat X
\right\|_{\rm tr}
&\leq
2p_{\rm coll}(m,R)\|\hat X\|_{\rm tr}
\nonumber\\
&\leq
\frac{m(m-1)}{R}\|\hat X\|_{\rm tr} .
\label{eq:site-map-bound}
\end{align}
\end{small}
Let us define
\begin{equation}\label{eq:Ft}
\begin{aligned}
F_t
&:=\sum_{\sigma,\tau}\int_0^t\int_0^t
|c^{\sigma\tau}(s,s')|\,\mathrm{d}s'\,\mathrm{d}s \\
&=4\int_0^t\int_0^t|c(s-s')|\,\mathrm{d}s'\,\mathrm{d}s.
\end{aligned}
\end{equation}
Since
\(\sum_{\sigma,\tau}|c^{\sigma\tau}(s,s')|
=4|c(s-s')|\), permutation symmetry of the pairing sum converts
the integral over \(\Delta_{2m}(t)\) into \(1/(2m)!\) times the
full-cube integral, which factorizes. Therefore,
\begin{small}
\begin{align*}
&\sum_{\pi}\int_{\Delta_{2m}(t)}\mathrm d\bm t
\sum_{\bm\sigma}
\prod_{\ell=1}^{m}
\left|c_\ell^{\bm\sigma}(\bm t)\right|
\nonumber =
\frac{(2m-1)!!}{(2m)!}F_t^m.
\end{align*}
\end{small}
\begin{small}
\begin{align}
\left\|\delta\hat\rho_{\text{s},I}^{(2m)}(t)\right\|_{\rm tr}
&\leq
\frac{m(m-1)}{R}
\frac{(2m-1)!!}{(2m)!}F_t^m
\nonumber\\
&=
\frac{m(m-1)}{R}\frac{F_t^m}{2^m m!}.
\label{eq:order-error-bound}
\end{align}
\end{small}
Unitary invariance and summation over all orders give
\nopagebreak[4]
\begin{small}
\begin{align*}
\left\|\hat\rho_{\text{s}}(t)-\hat\rho_{\text{s}}^{(R)}(t)\right\|_{\rm tr}
&\leq \sum_{m=2}^{\infty} \|\delta\hat\rho_{\text{s},I}^{(2m)}(t)\|_{\rm tr}
\nonumber\\
&\leq\frac{1}{R}
\sum_{m=2}^{\infty}m(m-1)\frac{(F_t/2)^m}{m!}
=\frac{A_t}{R},
\end{align*}
\end{small}
where \(A_t=(F_t^2/4)\mathrm e^{F_t/2}\). This constant depends only on the scalar bath correlation and the final time. No factor depending on \(N_{\text{s}}\) appears because the site sums have already been absorbed into the unitary phase-dressed operators \(\hat D_\varepsilon^{(a)}(t)\). The same argument applies to bounded electronic observables by duality.
\end{proof}

Next, we extend this result to the coupled-Lindblad realization used in our simulations.
For the coupled-Lindblad SPA model, define
\begin{small}
\begin{align*}
\hat B_{i,R}(\theta)
&=\frac{1}{\sqrt R}\sum_{a=1}^R\sum_{k=1}^{n_{\mathrm b}}
\left(r_i^{(a)}\epsilon_k\hat b_{a,k}^\dagger
+r_i^{(a)*}\epsilon_k^*\hat b_{a,k}\right),\\
\mathcal L_{\mathrm b}^{(R)}(\hat X)
&=-\mathrm i[\hat H_{\mathrm b}^{(R)},\hat X]+\mathcal D(\hat X),
\end{align*}\end{small}using the quantities defined in the main text, and let
\(\hat\rho_{\mathrm b}^{(R)}\) be the stationary Gaussian state of
\(\mathcal L_{\mathrm b}^{(R)}\). Let
\(\mathcal B_{i,R}^{\mathrm l}(\theta)\) and
\(\mathcal B_{i,R}^{\mathrm r}(\theta)\) denote left and right multiplication by
\(\hat B_{i,R}(\theta)\), respectively. For \(t_1\geq\cdots\geq t_n\), replace the unitary random-bath moment by
\begin{small}
\begin{align*}
\Phi_{\bm i,R}^{\bm\sigma}(\bm t)
&=\mathbb E_\theta\Tr_{\mathrm b^{(R)}}\!\Bigl[
\mathcal B_{i_1,R}^{\sigma_1}(\theta)
e^{(t_1-t_2)\mathcal L_{\mathrm b}^{(R)}}
\nonumber\\[-0.5ex]
&\qquad\cdots
e^{(t_{n-1}-t_n)\mathcal L_{\mathrm b}^{(R)}}
\mathcal B_{i_n,R}^{\sigma_n}(\theta)
\hat\rho_{\mathrm b}^{(R)}\Bigr].
\end{align*}\end{small}

\begin{corollary}[Coupled-Lindblad realization]\label{cor:lindblad}
\normalfont
Given a coupled-Lindblad $R$-bath SPA realization that reproduces the exact BCF $c(t)$, its phase-averaged reduced system state satisfies the conclusion of Theorem~\ref{thm:main} with the same system-size-independent constant $A_t$.
\end{corollary}

\begin{proof}
The only new point relative to Theorem~\ref{thm:main} is the bath propagation between interaction events. Exact reproduction of $c(t)$ gives, for $t\geq s$ and $\sigma,\tau\in\{\mathrm l,\mathrm r\}$,
\begin{equation*}
\mathbb E_\theta\Tr_{\mathrm b^{(R)}}\!\left[
\mathcal B_{i,R}^{\sigma}(\theta)e^{(t-s)\mathcal L_{\mathrm b}^{(R)}}
\mathcal B_{j,R}^{\tau}(\theta)\hat\rho_{\mathrm b}^{(R)}
\right]
=\delta_{ij}c^{\sigma\tau}(t,s).
\end{equation*}
The quadratic Hamiltonian and linear Lindblad operators preserve Gaussianity, so the Duhamel moments defined above obey Wick's theorem. The explicit phase dressing in $\hat B_{i,R}(\theta)$ then gives exactly the Wick expansion~\eqref{eq:wick-expansion}. From this point onward, the unitary proof applies verbatim: the repeated-channel estimate~\eqref{eq:collision-fraction} and site-map contraction~\eqref{eq:site-map-bound} yield the order-by-order bound~\eqref{eq:order-error-bound}, whose sum is controlled by $F_t$ in Eq.~\eqref{eq:Ft}.
\end{proof}
\end{document}

\preprint{APS/123-QED}
\newcommand{\pto}{PbTi$\text{O}_3$ }
\newcommand{\bto}{BaTi$\text{O}_3$ }
\newcommand{\red}[1]{\textcolor{red}{#1}}

\title{Supplemental Material for ``Scalable simulation of non-Markovian quantum transport by stochastic-phase bath reduction''}

\author{Zhen Huang}
\email{hertz@math.berkeley.edu}
\affiliation{Department of Mathematics, University of California, Berkeley, CA 94720, USA}
\affiliation{Center for Computational Quantum Physics, Flatiron Institute, New York, NY 10010, USA}
\affiliation{Center for Computational Mathematics, Flatiron Institute, New York, NY 10010, USA}
\author{Lin Lin}
\email{lin@caltech.edu}
\affiliation{Department of Mathematics, University of California, Berkeley, CA 94720, USA}
\affiliation{Applied Mathematics and Computational Research Division, Lawrence Berkeley National Laboratory, Berkeley, CA 94720, USA}
\affiliation{Department of Computing and Mathematical Sciences, California Institute of Technology, Pasadena, CA 91125, USA}
\author{Pinchen Xie}
\email{pinchenxie@lbl.gov}
\affiliation{Applied Mathematics and Computational Research Division, Lawrence Berkeley National Laboratory, Berkeley, CA 94720, USA}

\date{\today}

\maketitle

\section{Implementation of state evolution propagators}
Within the SPA framework, once the stochastic phase factors are sampled, the quantum dynamics can be propagated by any standard Hamiltonian-simulation technique, including full state-vector propagation, Krylov propagation~\cite{Saad1992}, and tensor-network time evolution~\cite{Vidal2004TEBD,WhiteFeiguin2004,Daley2004,Haegeman2011TDVP}. %The same Hamiltonian formulation is also compatible with digital quantum-simulation algorithms, including qDRIFT \cite{Campbell2019qDRIFT}  and qubitization approaches or quantum singular-value transformation~\cite{LowChuang2019Qubitization}.
For SPA simulations reported in the main text, we use full state-vector propagation with a second-order Suzuki--Trotter decomposition~\cite{Trotter1959,Suzuki1991}.

For a SPA simulation with $R$ global baths, each containing $n$ auxiliary bosonic modes, the total number of auxiliary modes is $Z=nR$. 
A full system-bath state vector is represented as a tensor $\Psi_{i,l_1,\ldots,l_{Z}}(t)$, where $i$ labels the tight-binding site and $l_z$ denotes the occupation of the auxiliary bosonic mode $z\in\{1,\ldots,Z\}$. We use $\theta$ to denote the sampled stochastic phase factors.

\subsection{Unitary propagator for uncompressed bath}
For uncompressed baths, the quantum dynamics are driven by standard unitary evolution. The Hamiltonian is $\hat H_\theta^{(R)}=\hat H_{\text{s}}+\hat H_{\text{b}}^{(R)}+\hat H_{\text{sb},\theta}^{(R)}$, where
\begin{equation}    
\hat H_{\text{b}}^{(R)}=\sum_{a=1}^{R}\sum_{\alpha=1}^{n}
\omega_\alpha \hat b_{a\alpha}^\dagger \hat b_{a\alpha},
\end{equation}
and
\begin{equation}\label{eq:ham-eb-eff}
\hat H_{\text{sb},\theta}^{(R)}
=
\sum_{i=1}^{N_{\text{s}}} |i\rangle\langle i|
\otimes
\sum_{a=1}^{R}\sum_{\alpha=1}^{n}
\frac{g_\alpha\omega_\alpha}{\sqrt R}
\left(
r_i^{(a)}\hat b_{a\alpha}^\dagger
+
r_i^{(a)*}\hat b_{a\alpha}
\right).
\end{equation}

For simplicity, we let $\hat{h}_{\text{b},z} \equiv \omega_\alpha \hat b_{a\alpha}^\dagger \hat b_{a\alpha}$ with $z=(a-1)n+\alpha$. For the same reason, we let $\hat{h}_{\text{sb},z} \equiv \sum_{i=1}^{N_{\text{s}}} |i\rangle\langle i|
\otimes
\frac{g_\alpha\omega_\alpha}{\sqrt R}
\left(
r_i^{(a)}\hat b_{a\alpha}^\dagger
+
r_i^{(a)*}\hat b_{a\alpha}
\right)$ with $z=(a-1)n+\alpha$. 

Because $\hat h_{\text{b},z}$ acts only on the auxiliary mode $z$, and $\hat h_{\text{sb},z}$ acts only on the system and the auxiliary mode $z$, we evolve the system-bath state vector using a second-order Trotter splitting scheme, given as 
\begin{equation}\label{sm:split}
\Psi(t+\delta t)
=
\hat U_{\text{s}}^{1/2}\left(\prod_{z=1}^Z \hat U_{\text{b},z}^{1/2}\right)
\left(\prod_{z=1}^Z \hat U_{\text{sb},z}\right)
\left(\prod_{z=1}^Z \hat U_{\text{b},z}^{1/2}\right)\hat U_{\text{s}}^{1/2}
\Psi(t).
\end{equation}
Here, 
$\hat U_{\text{s}}^{1/2}=e^{-\mathrm{i}\hat H_{\text{s}}\delta t/2}$, 
$\hat U_{\text{b},z}^{1/2}=e^{-\mathrm{i}\hat h_{\text{b},z}\delta t/2}$,
and $\hat U_{\text{sb},z}=e^{-\mathrm{i}\hat h_{\text{sb},z}\delta t}$. 

The operators $\hat U_{\text{s}}^{1/2}$, $\hat U_{\text{b},z}^{1/2}$, and $\hat U_{\text{sb},z}$ are calculated and stored before the simulation. Thus, the full many-body Hamiltonian is never formed.

\subsection{Lindblad propagator for compressed bath}

We next consider the case in which a local vibronic environment is compressed into $n$ auxiliary modes. With $R$ global baths, we still denote the total number of auxiliary modes by $Z=nR$. 

The quantum dynamics are driven by the Lindblad equation
\begin{equation}\label{sm:lindblad}
\frac{\mathrm{d}}{\mathrm{d}t}\hat\rho_{\theta}^{(R)}(t)=-\mathrm{i}[\hat H_{\text{s}}+\hat H_{\text{b}}^{(R)}+\hat H_{\text{sb},\theta}^{(R)},\hat\rho_{\theta}^{(R)}(t)]+\sum_{a=1}^{R}\sum_{k,l=1}^{n}
\Gamma_{kl}
\left(
2\hat b_{ak}\hat\rho_{\theta}^{(R)}(t)\hat b_{al}^\dagger
-
\{\hat b_{al}^\dagger\hat b_{ak},\hat\rho_{\theta}^{(R)}(t)\}
\right).
\end{equation}
By a slight abuse of notation, we let
\begin{equation}
\hat H_{\text{b}}^{(R)}=\sum_{a=1}^{R}\sum_{k,l=1}^{n}
K_{kl}\hat b_{ak}^\dagger\hat b_{al},
\end{equation}
and 
\begin{equation}
\hat H_{\text{sb},\theta}^{(R)}
=
\sum_{i=1}^{N_{\text{s}}} |i\rangle\langle i|
\otimes
\sum_{a=1}^{R}\sum_{k=1}^{n}
\frac{1}{\sqrt R}
\left(
r_i^{(a)}\epsilon_k\hat b_{ak}^\dagger
+
r_i^{(a)*}\epsilon_k^*\hat b_{ak}
\right).
\end{equation}
Here, the bath Hamiltonian matrix $K$, the damping matrix $\Gamma$, and the coupling vector $\epsilon$ are determined through the bath compression procedure.
During bath compression, either $K$ or $\Gamma$ can always be diagonalized, although they cannot generally be diagonalized simultaneously. In practice, we choose a bosonic basis in which $\Gamma=\mathrm{diag}\{\gamma_1,\ldots,\gamma_{n}\}$ is diagonal; we adopt this convention throughout this section. Equation~\ref{sm:lindblad} then becomes
 \begin{equation}    
\frac{\mathrm{d}}{\mathrm{d}t}\hat\rho_{\theta}^{(R)}(t)
=-\mathrm{i}[\hat H_{\text{s}}+\hat H_{\text{b}}^{(R)}+\hat H_{\text{sb},\theta}^{(R)},\hat\rho_{\theta}^{(R)}(t)]
+\sum_{a=1}^{R}\sum_{k=1}^{n}
\gamma_{k}
\left(
2\hat b_{ak}\hat\rho_{\theta}^{(R)}(t)\hat b_{ak}^\dagger
-
\{\hat b_{ak}^\dagger\hat b_{ak},\hat\rho_{\theta}^{(R)}(t)\}
\right).
\end{equation}

Instead of directly evolving the system-bath density operator $\hat\rho_{\theta}^{(R)}$, we adopt the Monte Carlo wavefunction propagator~\cite{PlenioKnight1998} that unravels the Lindblad equation into stochastic trajectories of the pure system-bath state vector $\Psi$ subject to random quantum jumps. This propagator is efficient for SPA because SPA also requires ensemble averaging over stochastic realizations. Each trajectory independently samples both the phase realization $\theta$ and a quantum-jump history, so a single joint ensemble average accounts for both sources of stochasticity.

Within the Monte Carlo wavefunction framework, the dense bath Hamiltonian $\hat H_{\text{b}}^{(R)}$ and jump operators \(\hat L_z\) define the effective non-Hermitian Hamiltonian
\begin{equation*}
    \hat H_{\text{eff}}
    =
    \hat H_{\text{s}}
    +
    \hat H_{\text{b}}^{(R)}
    +
    \hat H_{\text{sb},\theta}^{(R)}
    -
    \frac{\mathrm{i}}{2}\sum_{z=1}^Z \hat L_z^\dagger \hat L_z.
\end{equation*}
Here, $\hat L_z$ is given by $\hat L_z = \sqrt{2\gamma_k}\hat b_{ak}$, where $z=(a-1)n+k$.
For each time step, the pure state $\Psi$ is first propagated by $\hat H_{\text{eff}}$ through a Trotter splitting scheme over $\delta t$. Because of the non-Hermitian term in $\hat H_{\text{eff}}$, this propagation is not norm-conserving.
The loss of norm in $\Psi$ gives the jump probability; if a quantum jump occurs, the channel is sampled with weight \(\|\hat L_z|\Psi\rangle\|^2\), the corresponding \(\hat L_z\) is applied, and then the state is normalized. The details of such a scheme can be found in Ref.~\cite{PlenioKnight1998}. 

A Python implementation of these propagators can be found in the source code of \texttt{scalabath}~\cite{scalabath}.

\newpage

\section{Typical organic materials in the intermediate-coupling regime}

To illustrate how frequently the intermediate-coupling regime occurs in
organic quantum transport, Table~\ref{tab:intermediate-coupling-materials}
collects representative parameters for molecular semiconductors and
photosynthetic pigment--protein complexes.  We define $|V|$ as the largest
reported magnitude of the electronic or excitonic coupling among the relevant
intermolecular pathways and $\lambda$ as the corresponding one-site local
reorganization energy.  When a reference reports the total two-molecule
self-exchange reorganization energy, we divide that value by two to obtain the
local relaxation energy entering a Holstein model.  Reported polaron binding
energies and site reorganization energies obtained from spectral densities are
already local and are therefore used directly.  Values originally given in
$\mathrm{cm}^{-1}$ are converted to meV.

\begin{table}[h]
    \centering
    \begin{ruledtabular}
    \begin{tabular}{lccccc}
        Material or complex & Carrier & $|V|$ (meV) & $\lambda$ (meV) & 
        $\lambda/|V|$ & Reference \\
        Rubrene (single crystal) & hole & 83.0 & 80.0 & 0.96 & \cite{Coropceanu2007ChemRev} \\
        Pentacene & hole & 85.0 & 49.0 & 0.58 & \cite{Coropceanu2007ChemRev} \\
        Anthracene & hole & 44.0 & 69.0 & 1.57 & \cite{Coropceanu2007ChemRev} \\
        Naphthalene & hole & 35.0 & 93.0 & 2.66 & \cite{Coropceanu2007ChemRev} \\
        Tetracene & hole & 70.0 & 57.0 & 0.81 & \cite{Coropceanu2007ChemRev} \\
        DNTT & hole & 91.0 & 65.0 & 0.71 & \cite{Yamamoto2010DNTT} \\
        C$_7$-BTBT-C$_7$ & hole & 51.1 & 122.0 & 2.39 & \cite{Pandey2023BTBT} \\
        C$_8$-BTBT-C$_8$ & hole & 45.2 & 121.8 & 2.69 & \cite{Pandey2023BTBT} \\
        C$_{12}$-BTBT-C$_{12}$ & hole & 65.1 & 121.6 & 1.87 & \cite{Pandey2023BTBT} \\
        Monoclinic [60]PCBM & electron & 49.21 & 69.5 & 1.41 & \cite{Gajdos2013PCBM} \\
        FMO complex & exciton & 10.9 & 4.3 & 0.39 & \cite{Ishizaki2009PNAS} \\
        LH2 B850 ring (\textit{R. molischianum}) & exciton & 45 & 27 & 0.60 & \cite{kundu2020real} \\
        LHCII & exciton & 16.0 & 17.5 & 1.09 & \cite{Zhu2025LHCII} \\
        PC645 & exciton & 39.6 & 112.7 & 2.85 & \cite{Blau2018PC645} \\
        CP47 & exciton & 12.3 & 4.8 & 0.39 & \cite{Raszewski2008PSII,Kreisbeck2016PSII} \\
    \end{tabular}
    \end{ruledtabular}
    \caption{Representative electronic couplings and local reorganization
    energies for organic semiconductors and photosynthetic complexes.  The
    values are literature-derived estimates selected to compare the dominant
    coupling scale with the local vibronic relaxation scale; they are not
    unique material constants.}
    \label{tab:intermediate-coupling-materials}
\end{table}

Across these examples, $\lambda/|V|$ ranges from approximately $0.4$ to $2.8$.
Thus, the coherent coupling and vibronic relaxation energies generally differ
by no more than a factor of a few, placing these materials in or near the crossover
where neither weak- nor strong-coupling limits are uniformly reliable.  This
classification is an order-of-magnitude guide rather than a sharp boundary:
temperature, bandwidth, and dimensionality, static disorder, the distribution
of vibrational frequencies, and nonlocal electron--phonon coupling can all
affect the operative transport regime.  The tabulated parameters also depend
on crystal structure, transport direction, and electronic-structure or
spectroscopic parametrization.  For example, the rubrene simulations in the
main text use the model-specific value $\lambda=73~\mathrm{meV}$ rather than
the representative $80~\mathrm{meV}$ literature value listed here.

\newpage

\section{Details of simulations of the square lattice}

\subsection{$3\times 3$ lattice}

The calculation in Fig.~2(a,b) of the main text uses a single-particle
Holstein model on an open $3\times 3$ square lattice with unit lattice spacing.
All site energies are set to zero, and the hopping is isotropic,
$V_x=V_y=V=100~\mathrm{cm}^{-1}$.  Each site is coupled to one harmonic mode
of frequency $\omega_1=1000~\mathrm{cm}^{-1}$.
The dimensionless coupling $g_1$ is selected through
$\lambda=g_1^2\omega_1$ for $\lambda/V=0.5,1,2,$ and $4$.  The initial electronic
state is localized at the central site $(x,y)=(1,1)$ in zero-based lattice
coordinates, and the vibrational modes are initially thermal at $T=300~\mathrm{K}$.

Each bosonic mode is truncated to keep only the $d=6$ lowest oscillator levels, which yields nearly converged results for the simulations considered in the main text.  
For time evolution, the initial bath state is a pure product number state over all bosonic modes, sampled with independent local
probabilities $p_l=\frac{e^{-\beta\omega_1 l}}{\sum_{m=0}^{d-1}e^{-\beta\omega_1 m}}$. 
Here, $l=0,\ldots,d-1$, and the resulting pure system-bath state vector is propagated unitarily.  
The total state-vector dimension is $9\times d^9=90\,699\,264$.  

For each value of $\lambda/V$, the ``exact'' results in Fig.~2 of the main text are obtained by averaging over 96 independent trajectories whose initial thermal bath states are sampled from the same distribution.
Here, ``exact'' means direct unitary propagation within the specified oscillator
cutoff and time discretization. 

The corresponding SPA simulations replace the nine local bosonic modes by
$R$ shared auxiliary modes and use the phase-dressed coupling defined in the main
text.  Results are shown for $R=1$ and $R=3$.  For every value of $\lambda/V$ and every $R$, the SPA results are obtained by averaging over 200 independent trajectories with initial thermal bath state sampled from the same distribution. 
The state-vector dimension per trajectory is only $9d=54$ for $R=1$ and $9d^3=1944$ for $R=3$.

Both the direct and SPA state vectors are propagated with the second-order Trotter splitting scheme, using $\delta t=0.1~\mathrm{fs}$.  Populations are saved every $1~\mathrm{fs}$ up to $500~\mathrm{fs}$, and the production datasets use single-precision complex arithmetic.  After tracing out the bath, we calculate
\begin{equation*}
    P_{\mathrm{center}}(t)=P_{(1,1)}(t),
    \qquad
    \mathrm{MSD}(t)=\sum_i P_i(t)
    \left[(x_i-\bar x)^2+(y_i-\bar y)^2\right],
\end{equation*}
where $\bar x=\sum_i x_iP_i(t)$ and $\bar y=\sum_i y_iP_i(t)$.  The MSD is
therefore reported in units of the squared lattice spacing.

For comparison, we also propagate the reduced electronic density matrix with
the second-order time-convolutionless equation
(TCL2)~\cite{Fetherolf2017TCL2}.  The bath correlation used by TCL2 is matched
to the same truncated oscillator as the state-vector calculations.
The TCL2 equation is integrated deterministically with a fourth-order
Runge-Kutta method on the same $0.1~\mathrm{fs}$ grid. The plotted TCL2
populations are neither clipped nor renormalized. Consequently, the strong
negative eigenvalues and populations that develop at $\lambda/V=4$ remain
visible and diagnose the loss of positivity of the perturbative dynamics.

\subsection{$100\times 100$ lattice}

The calculation in Fig.~2(c,d) of the main text uses an open
$100\times 100$ lattice with the same $V=100~\mathrm{cm}^{-1}$,
$\omega=1000~\mathrm{cm}^{-1}$, and $T=300~\mathrm{K}$.  Independent static
site energies are drawn for each trajectory from
\begin{equation*}
    U_{(i,j)}\sim\operatorname{Uniform}[-W,W],
    \qquad W=200~\mathrm{cm}^{-1}.
\end{equation*}

The particle is initially localized at site $(49,49)$ in zero-based lattice
coordinates.  Production SPA simulations were performed for
$\lambda/V=8$ using $R=1$ and a $d=6$ cutoff for the auxiliary modes.
Each dataset contains 100 trajectories, with an independently
sampled disorder realization, set of stochastic phase factors, and thermal bath state for each trajectory.
Double-precision complex arithmetic is used here as required to describe long-time diffusive dynamics on the large lattice.  
The dynamics are propagated from $0$ to $10~\mathrm{ps}$ with $\delta t=0.1~\mathrm{fs}$, and
the populations are saved every $10~\mathrm{fs}$.  The main-text snapshots show the ensemble-averaged populations at $1$ and
$10~\mathrm{ps}$ on a logarithmic color scale.

As a qualitative strong-coupling comparison, we also propagate a
Lang--Firsov renormalized tight-binding (RTB) Hamiltonian for each of the same
100 disorder realizations. The RTB Hamiltonian is a plain tight-binding model with a renormalized nearest-neighbor hopping
$V'=V\exp(-g_1^2\coth(\beta\omega_1/2))$. There is no thermal bath in this model. 
Its dynamics therefore consists of strictly unitary tight-binding evolution.

At $\lambda/V=8$, $V'=44.34~\mathrm{cm}^{-1}$.  This RTB model is not a systematically converged
reference. It is included only as a qualitative strong-coupling baseline.  
The comparison in Fig.~2(d) shows that SPA and RTB dynamics
both predict strongly suppressed, disorder-frustrated spreading on the
$10~\mathrm{ps}$ time scale.

\section{Details of simulations of BChl aggregates}
\subsection{Model Parameters}

The BChl chain calculations follow the 19-site benchmark model of Ref.~\cite{kundu2020real}.  In the homogeneous case the common site energy is removed as a global energy shift and the nearest-neighbor hopping is $363~\mathrm{cm}^{-1}$.  In the inhomogeneous (staggered) case the one-based site energies alternate as
\begin{equation*}
    U_i=
    \begin{cases}
        12653.66~\mathrm{cm}^{-1}, & i\ \mathrm{odd},\\
        12457.66~\mathrm{cm}^{-1}, & i\ \mathrm{even},
    \end{cases}
\end{equation*}
and the common offset $12457.66~\mathrm{cm}^{-1}$ is subtracted in the simulation.  The nearest-neighbor hoppings alternate as $363$ and $320~\mathrm{cm}^{-1}$, beginning with the bond between sites 1 and 2.  In both cases the initial excitation is localized on site 10.  For the inhomogeneous case, site 10 is the lower-energy site.

Each BChl site is coupled to 50 intramolecular modes.  We denote the Huang--Rhys factor by $S_m=g_m^2$, so that the local reorganization energy is
\begin{equation*}
    \lambda=\sum_{m=1}^{50}S_m\omega_m
    =217.66~\mathrm{cm}^{-1}.
\end{equation*}
The mode frequencies and Huang--Rhys factors are listed in Table~\ref{tab:bchl-bath}.

\begin{table}[h]
    \centering
    \scriptsize
    \begin{ruledtabular}
    \begin{tabular}{ccc ccc ccc ccc ccc}
        $m$ & $\omega_m$ & $S_m$ & $m$ & $\omega_m$ & $S_m$ & $m$ & $\omega_m$ & $S_m$ & $m$ & $\omega_m$ & $S_m$ & $m$ & $\omega_m$ & $S_m$ \\
        1 & 84 & 0.0151 & 11 & 407 & 0.0052 & 21 & 710 & 0.0018 & 31 & 1001 & 0.0040 & 41 & 1211 & 0.0025 \\
        2 & 167 & 0.0081 & 12 & 423 & 0.0029 & 22 & 727 & 0.0266 & 32 & 1019 & 0.0097 & 42 & 1229 & 0.0021 \\
        3 & 183 & 0.0072 & 13 & 442 & 0.0023 & 23 & 776 & 0.0095 & 33 & 1066 & 0.0025 & 43 & 1252 & 0.0025 \\
        4 & 191 & 0.0196 & 14 & 473 & 0.0017 & 24 & 803 & 0.0042 & 34 & 1089 & 0.0021 & 44 & 1289 & 0.0087 \\
        5 & 214 & 0.0046 & 15 & 506 & 0.0016 & 25 & 845 & 0.0025 & 35 & 1105 & 0.0021 & 45 & 1378 & 0.0086 \\
        6 & 239 & 0.0078 & 16 & 565 & 0.0081 & 26 & 858 & 0.0025 & 36 & 1117 & 0.0103 & 46 & 1466 & 0.0021 \\
        7 & 256 & 0.0055 & 17 & 587 & 0.0039 & 27 & 890 & 0.0284 & 37 & 1137 & 0.0042 & 47 & 1519 & 0.0017 \\
        8 & 345 & 0.0161 & 18 & 623 & 0.0058 & 28 & 915 & 0.0048 & 38 & 1158 & 0.0103 & 48 & 1539 & 0.0023 \\
        9 & 368 & 0.0060 & 19 & 684 & 0.0023 & 29 & 967 & 0.0027 & 39 & 1180 & 0.0025 & 49 & 1648 & 0.0025 \\
        10 & 388 & 0.0041 & 20 & 696 & 0.0025 & 30 & 980 & 0.0031 & 40 & 1190 & 0.0031 & 50 & 1680 & 0.0027 \\
    \end{tabular}
    \end{ruledtabular}
    \caption{BChl bath frequencies $\omega_m$ in $\mathrm{cm}^{-1}$ and Huang--Rhys factors $S_m$.}
    \label{tab:bchl-bath}
\end{table}

\subsection{BCF fitting results}

At $T=300~\mathrm{K}$ the original 50-mode thermal BCF is
\begin{equation*}
    c(t)=\sum_{m=1}^{50}S_m\omega_m^2
    \left[
    \coth\left(\frac{\beta\omega_m}{2}\right)\cos(\omega_m t)
    -\mathrm{i}\sin(\omega_m t)
    \right].
\end{equation*}
This BCF is compressed into a six-mode coupled-Lindblad representation,
\begin{equation*}
    c_{\text{fit}}(t)=
    \epsilon^\dagger
    \exp[(-\mathrm{i} K-\Gamma)t]
    \epsilon.
\end{equation*}
$K$ is a dense $6\times 6$ Hermitian matrix, $\epsilon$ is a $6$-vector, and $\Gamma=\mathrm{diag}(\gamma_1,\ldots,\gamma_6)$ is a diagonal matrix. These parameters are optimized by minimizing the error between the original and fitted BCFs for $t\in[0,100~\mathrm{fs}]$. A comparison between the original BCF and the fitted BCF is shown in Supplemental Fig.~\ref{fig:compare-bcf}.

\begin{figure}[h]
    \centering
    \includegraphics[width=0.7\textwidth]{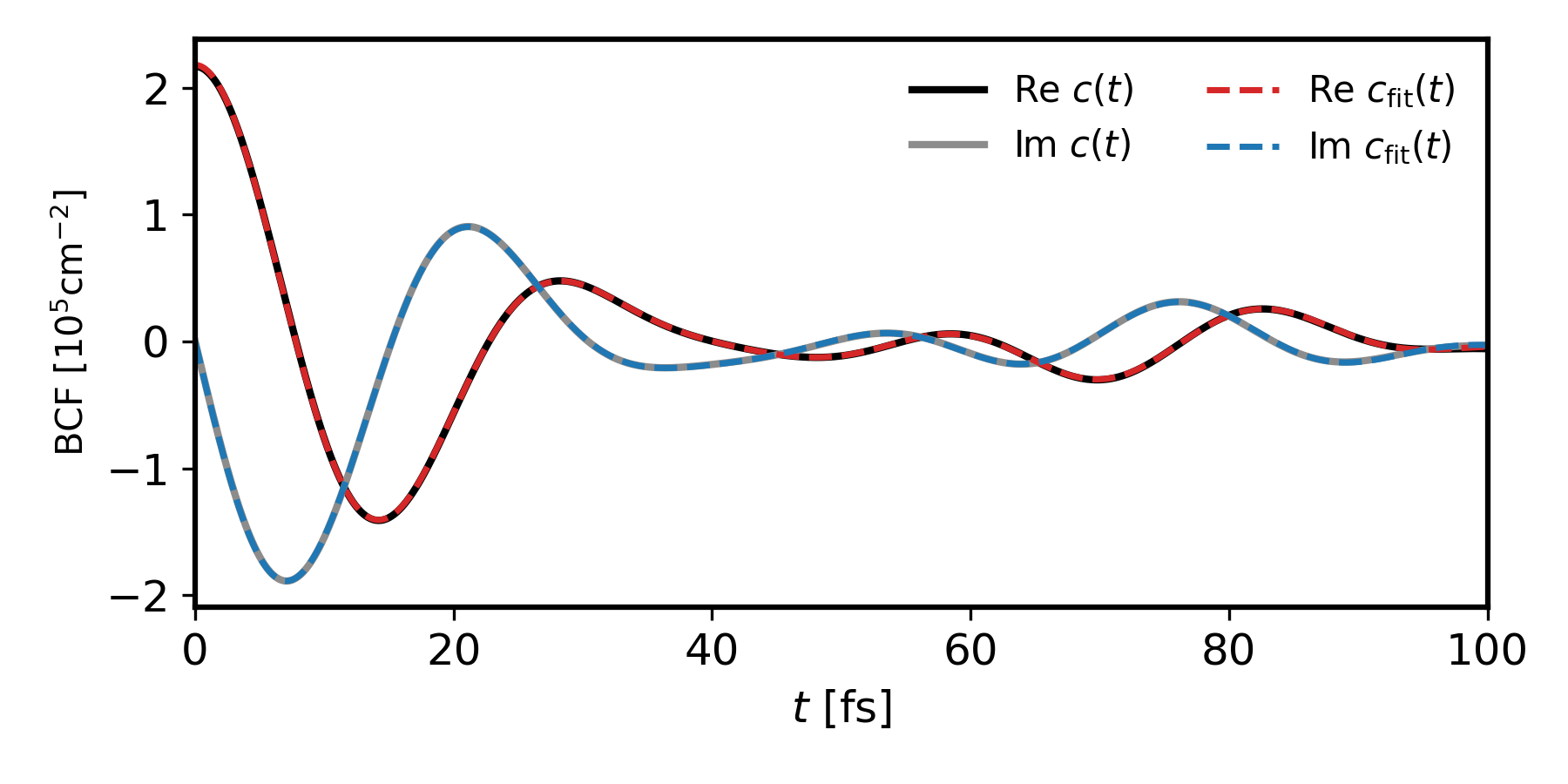}
    \caption{Comparison between the original BCF and the fitted BCF.}
    \label{fig:compare-bcf}
\end{figure}

The optimal parameters are:
\begin{equation*}
\begin{aligned}
\epsilon &=
(12.189,\ 238.799,\ 103.753,\ 360.181,\ 85.506,\ 114.003)~\mathrm{cm}^{-1},\\
\gamma &=
(873.185,\ 0,\ 0,\ 0,\ 0,\ 0)~\mathrm{cm}^{-1}.
\end{aligned}
\end{equation*}

Writing $K=K_{\text{R}}+\mathrm{i} K_{\text{I}}$, the optimal $K_{\text{R}}$ and $K_{\text{I}}$ are:
\begin{equation*}
K_{\text{R}}=
\begin{pmatrix}
318.03 & -238.13 & -205.51 & 611.88 & -180.31 & 123.05 \\
-238.13 & -76.10 & 294.49 & 570.49 & 88.62 & 133.35 \\
-205.51 & 294.49 & 1131.76 & 108.90 & 126.29 & -395.79 \\
611.88 & 570.49 & 108.90 & 601.36 & -64.34 & 76.93 \\
-180.31 & 88.62 & 126.29 & -64.34 & 602.00 & 150.21 \\
123.05 & 133.35 & -395.79 & 76.93 & 150.21 & 318.22
\end{pmatrix}
\end{equation*}
and
\begin{equation*}
K_{\text{I}}=
\begin{pmatrix}
0.00 & -69.88 & -398.04 & 266.19 & -97.23 & -121.94 \\
69.88 & 0.00 & 198.31 & -90.66 & -15.76 & -262.04 \\
398.04 & -198.31 & 0.00 & -155.34 & 202.20 & 512.79 \\
-266.19 & 90.66 & 155.34 & 0.00 & 156.43 & 77.03 \\
97.23 & 15.76 & -202.20 & -156.43 & 0.00 & -98.37 \\
121.94 & 262.04 & -512.79 & -77.03 & 98.37 & 0.00
\end{pmatrix}.
\end{equation*}
All entries of $K_{\mathrm R}$ and $K_{\mathrm I}$ are in $\mathrm{cm}^{-1}$.
The BCF fitting is numerically accurate. The fitting error is negligible on the $100~\mathrm{fs}$ benchmark window compared with the SPA simulation error from using $R=1$.

\subsection{SPA simulations}

The compressed BChl bath is shared by all sites with $R=1$. The six auxiliary modes are truncated with local dimensions
\begin{equation*}
    (5,5,5,5,5,5),
\end{equation*}
giving a global-bath Hilbert-space dimension of $15\,625$.

For each trajectory propagated by the SPA simulation, stochastic phase factors are sampled independently, and the initial bath state is fixed to 
\begin{equation*}
    |0,0,0,0,0,0\rangle .
\end{equation*}
Although the auxiliary modes are initialized in their vacuum, the fitted coupled-Lindblad parameters ensure that their BCF reproduces the thermal BCF of the original $300~\mathrm{K}$ vibrational environment.

The BChl simulations use the Monte Carlo wavefunction propagator with a time step of $\delta t=0.1~\mathrm{fs}$ and single-precision complex arithmetic. Populations are saved every $1~\mathrm{fs}$, and the final time is $100~\mathrm{fs}$. The results reported in Fig.~3(b) of the main text are obtained by averaging over 128 independently propagated trajectories for each of the homogeneous and inhomogeneous chains.

\newpage
\section{Details of SPA simulations of Rubrene crystal}
\subsection{Model Parameters}

The rubrene calculations use the one-dimensional Holstein model described in the main text.  The constant site energy is omitted, the nearest-neighbor hopping is $V=83~\mathrm{meV}$, and each molecular site is coupled to nine bosonic modes.  The mode frequencies $\omega_m$, partial reorganization energies $\lambda_m$, and dimensionless couplings $g_m$ are listed in Table~\ref{tab:couple}.  The couplings are computed as $g_m=(\lambda_m/\omega_m)^{1/2}$, so that the total local reorganization energy is $\lambda=\sum_{m=1}^9 \lambda_m\approx 73~\mathrm{meV}$.

\begin{table}[h]
    \centering
    \begin{ruledtabular}
    \begin{tabular}{cccccccccc}
        $m$ & 1 & 2 & 3 & 4 & 5 & 6 & 7 & 8 & 9\\
        $\omega_m / \mathrm{cm}^{-1}$  & 84.000 & 214.133 & 631.692 & 1002.141 & 1205.567 & 1351.178 & 1364.026 & 1535.332 & 1594.000 \\
        $\lambda_m / \mathrm{cm}^{-1}$ & 77.797 & 29.492 & 39.661 & 39.661 & 27.458 & 131.695 & 23.390 & 63.559 & 155.593 \\
        $g_m$ & 0.962 & 0.371 & 0.251 & 0.199 & 0.151 & 0.312 & 0.131 & 0.204 & 0.312 \\
    \end{tabular}
    \end{ruledtabular}
    \caption{Bath parameters for rubrene used in the simulations.}
    \label{tab:couple}
\end{table}

\begin{figure}[b]
    \centering
    \includegraphics[width=0.8\textwidth]{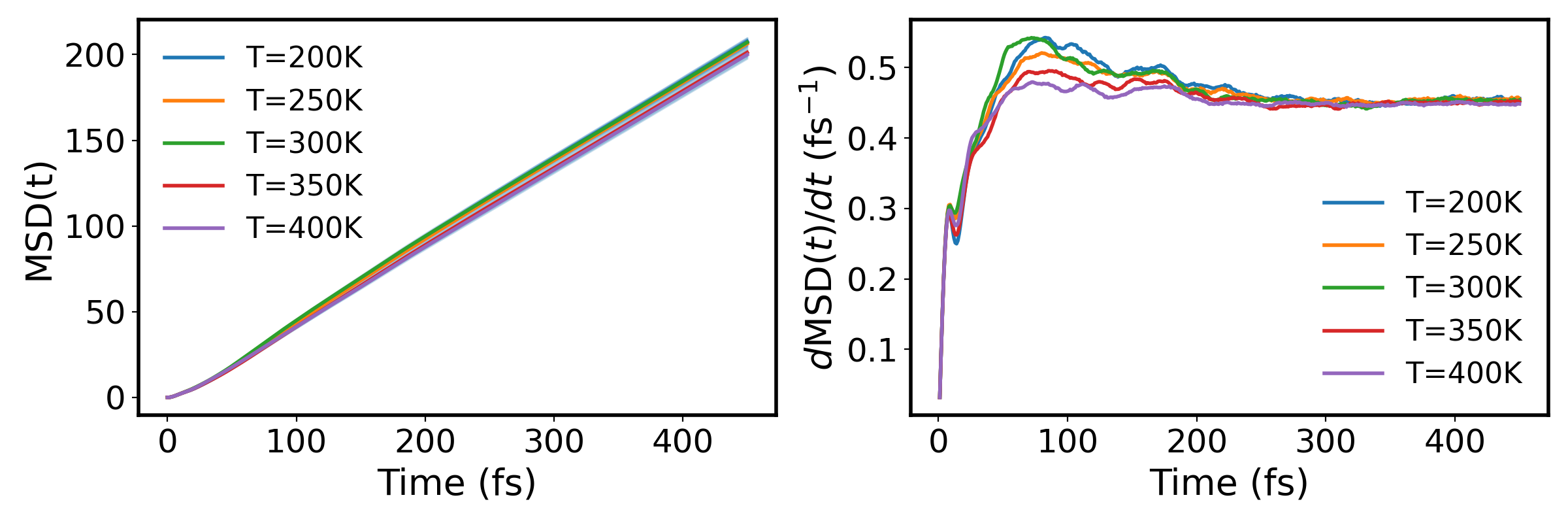}
    \caption{MSD and the slope of the MSD as functions of time for the SPA rubrene simulation.}
    \label{fig:sm-msd-comparison}
\end{figure}

\subsection{SPA simulations}

Because the rubrene bath contains only nine undamped modes, bath compression is not necessary.  The SPA calculation uses a single shared copy of the nine-mode bath ($R=1$).  The nine auxiliary modes are truncated with local dimensions
\begin{equation*}
    (12,6,4,3,3,4,3,3,4),
\end{equation*}
corresponding to a global-bath Hilbert-space dimension of $373\,248$.
 
The production SPA simulations use $L=200$, open boundary conditions, single-precision complex arithmetic, $\delta t=0.1~\mathrm{fs}$, and samples saved every $1~\mathrm{fs}$ up to $450~\mathrm{fs}$.  The initial electronic state is localized on the center site, $|100\rangle$ in one-based indexing, and the initial bath state is sampled from the thermal distribution at the target temperature.  For the mobility calculation we use $T=200,250,300,350,$ and $400~\mathrm{K}$, with 100 independently propagated trajectories per temperature.

The position operator used to define the mean-squared displacement is the site-index displacement operator
\begin{equation*}
    \hat n=\sum_i (i-i_0)|i\rangle\langle i|,
\end{equation*}
where $i_0=100$ is the initial site and $\bar\rho_{\mathrm{s}}=M^{-1}\sum_r\hat\rho_{\mathrm{s},r}$ is the ensemble-averaged state. Before boundary effects, reflection symmetry gives $\operatorname{Tr}[\bar\rho_{\mathrm{s}}(t)\hat n]=0$, so
\begin{equation*}
    \mathrm{MSD}(t)=\operatorname{Tr}[\bar\rho_{\mathrm{s}}(t)\hat n^2]
    =\frac{1}{M}\sum_{r=1}^{M}\operatorname{Tr}[\hat\rho_{\mathrm{s},r}(t)\hat n^2].
\end{equation*}
The plotted mobility is obtained from the ensemble-averaged MSD through
\begin{equation*}
    \mu_{\text{b}}(T)=
    \frac{eD^2}{2k_{\mathrm B}T}
   \lim_{t\to\infty} \frac{\mathrm{d}}{\mathrm{d}t}\mathrm{MSD}(t).
\end{equation*}
Here, $e$ is the elementary charge, and $D=7.19~\text{\AA}$ is the intermolecular distance.

The MSD and the slope of the MSD, obtained by averaging over 100 trajectories, are reported in Supplemental Fig.~\ref{fig:sm-msd-comparison}. The slope reaches an approximately constant plateau over the diffusive time window before finite-size effects limit the spreading. We use the average slope of the MSD over the interval $300$--$350~\mathrm{fs}$ to represent the long-time limit of $\frac{\mathrm{d}}{\mathrm{d}t}\mathrm{MSD}(t)$.

\subsection{Ehrenfest simulations}

The Ehrenfest curve shown in Fig.~3(d) of the main text was generated for the
same $L=200$ rubrene chain, hopping $V=83~\mathrm{meV}$, nine-mode bath, and
temperature range used in the SPA calculation.  In this mean-field
mixed quantum--classical simulation~\cite{ehrenfest1927bemerkung}, every site
is coupled self-consistently to nine classical harmonic oscillators with the
frequencies and couplings in Table~\ref{tab:couple}.  The oscillator positions
and velocities are sampled from the classical thermal distribution at the
target temperature; at the initially occupied site, the mean initial position
is displaced to the corresponding polaron equilibrium position.   

For each temperature, 128 independent trajectories are propagated with a
$0.02~\mathrm{fs}$ time step for $450~\mathrm{fs}$, and the populations and MSD
are recorded every $1~\mathrm{fs}$.  The MSD is evaluated separately for every
trajectory using the definition above and then ensemble averaged.  The
long-time limit of the slope of the MSD is obtained by averaging the
$300$--$450~\mathrm{fs}$ interval.  The mobility is calculated with $D=7.19~\text{\AA}$.

\subsection{DIQCD simulations}

The data-informed quantum--classical dynamics (DIQCD) model is trained
separately at each temperature against quantum dynamics of an effective two-level system
coupled to the nine-mode rubrene bath~\cite{Xie2026DIQCD}.  These
training data are generated through unitary dynamics with the same per-mode bosonic cutoffs as the SPA
calculation,
\begin{equation*}
    (12,6,4,3,3,4,3,3,4),
\end{equation*}
corresponding to a local-bath Hilbert-space dimension of $373\,248$. The time step is $0.1~\mathrm{fs}$, and each unitary trajectory is saved at $1~\mathrm{fs}$ intervals up to $100~\mathrm{fs}$.

\begin{table}[b]
    \centering
    \begin{ruledtabular}
    \begin{tabular}{ccccccc}
        $T$ (K) & TD-DMRG & SPA & FGR & Boltzmann & DIQCD & Ehrenfest \\
        200 & 65.864 & 67.629 & 88.409 & 160.084 & 65.075 & 242.232 \\
        250 & 53.116 & 54.232 & 66.875 & 114.286 & 59.594 & 157.520 \\
        300 & 46.034 & 44.612 & 52.779 & 88.235 & 46.418 & 112.301 \\
        350 & 38.102 & 38.294 & 42.942 & 71.008 & 33.126 & 88.534 \\
        400 & 33.569 & 33.489 & 35.750 & 57.983 & 35.049 & 68.467 \\
        \hline
        MaxAE &   & 1.764 & 22.545 & 94.220 & 6.478 & 176.368 \\
    \end{tabular}
    \end{ruledtabular}
    \caption{The rubrene $b$-axis mobilities in $\mathrm{cm}^2\,\mathrm{V}^{-1}\,\mathrm{s}^{-1}$.  The final row gives the maximum absolute error (MaxAE) relative to TD-DMRG over the five temperatures.}
    \label{tab:rubrene-mobility-reference}
\end{table}

The training data are obtained by averaging over 128 independently sampled trajectories. 
At each temperature, the DIQCD model parameters are optimized against the ensemble mean and standard deviation of observables of the effective two-level system. The choice of the DIQCD model parameters and the training procedure are consistent with those in Ref.~\cite{Xie2026DIQCD}. Details are omitted here.

The trained DIQCD model is then applied to the $L=200$ rubrene chain using 512
independent trajectories, a $0.05~\mathrm{fs}$ propagation step, and a
$1~\mathrm{fs}$ output interval for the same $450~\mathrm{fs}$ duration.  The
MSD and mobility are evaluated exactly as for the Ehrenfest data: the
finite-difference MSD slope is averaged over the
$300$--$450~\mathrm{fs}$ time interval. $D=7.19~\text{\AA}$ is
used in the mobility conversion.

The setup of Ehrenfest and DIQCD simulations in this work differs from that of
Ref.~\cite{Xie2026DIQCD}, where shorter propagation times, smaller bosonic
cutoffs, and $D=7~\text{\AA}$ led to overestimated mobilities. In particular, the DIQCD mobility reported there was calculated using the finite-difference slope of the MSD over the interval $50$--$100~\mathrm{fs}$, which leads to an unconverged estimate of the long-time limit of the slope. With these parameters corrected, DIQCD agrees substantially better with TD-DMRG than previously reported.

The numerical values of the mobilities reported in Fig.~3(d) of the main text are listed in Table~\ref{tab:rubrene-mobility-reference}. The maximum absolute error (MaxAE) relative to TD-DMRG is also reported. Among all methods, SPA achieves the smallest MaxAE and agrees consistently with TD-DMRG.

\bibliographystyle{apsrev4-2}
\bibliography{references}